\documentclass{article}

\usepackage[T1]{fontenc}
\usepackage[latin9]{inputenc}
\usepackage{amsmath, amsthm, amssymb}
\usepackage{graphicx}
\usepackage{enumitem}
\usepackage{natbib}
\usepackage{bussproofs}
\usepackage[letterpaper, verbose, tmargin=1in, bmargin=1in, lmargin=1in, rmargin=1in]{geometry}
\usepackage{appendix}
\usepackage{mathrsfs}
\usepackage{thmtools, thm-restate}
\usepackage{tikz}
\usepackage{pgf}
\usetikzlibrary{patterns, automata, arrows, shapes, snakes, topaths, trees, backgrounds, positioning, through, calc}
\usepackage{setspace}
\usepackage{multicol}
\usepackage{stmaryrd}
\usepackage{hyperref}
\hypersetup{colorlinks=true, citecolor=black, linkcolor=black}
\usepackage{comment}
\setcitestyle{round}
\usepackage[english]{babel} 
\definecolor{medgreen}{rgb}{0.0, 0.75, 0.0}
\definecolor{darkgreen}{rgb}{0.0, 0.35, 0.0}

\theoremstyle{definition}
\newtheorem{theorem}{Theorem}[section]

\newtheorem{proposition}[theorem]{Proposition}
\newtheorem{corollary}[theorem]{Corollary}
\newtheorem{definition}[theorem]{Definition}

\newtheorem{lemma}[theorem]{Lemma}
\newtheorem{fact}[theorem]{Fact}
\newtheorem{example}[theorem]{Example}

\newtheorem{remark}[theorem]{Remark}

\newcommand{\maximum}{m}

\begin{document}

\onehalfspace

\title{A partial-state space model of unawareness}

\author{Wesley H. Holliday \\ University of California, Berkeley}

\date{{\small Forthcoming in \textit{Journal of Mathematical Economics}}}

\maketitle

\begin{abstract}We propose a model of unawareness that remains close to the paradigm of Aumann's model for knowledge [R.~J.~Aumann, International Journal of Game Theory 28 (1999) 263-300]: just as Aumann uses a correspondence on a state space to define an agent's knowledge operator on events, we use a correspondence on a state space to define an agent's awareness operator on events. This is made possible by three ideas. First, like the model of [A.~Heifetz, M.~Meier, and B.~Schipper, Journal of Economic Theory 130
(2006) 78-94], ours is based on a space of partial specifications of the world, partially ordered by a relation of further specification or refinement, and the idea that agents may be aware of some coarser-grained specifications while unaware of some finer-grained specifications; however, our model is based on a different implementation of this idea, related to forcing in set theory. Second, we depart from a tradition in the literature, initiated by [S.~Modica and A.~Rustichini, Theory and Decision 37 (1994) 107-124] and adopted by Heifetz et al. and [J.~Li, Journal of Economic Theory 144 (2009) 977-993], of taking awareness to be definable in terms of knowledge. Third, we show that the negative conclusion of a well-known impossibility theorem concerning unawareness in [Dekel, Lipman, and Rustichini, Econometrica 66 (1998) 159-173] can be escaped by a slight weakening of a key axiom. Together these points demonstrate that a correspondence on a partial-state space is sufficient to model unawareness of events. Indeed, we prove a representation theorem showing that any abstract Boolean algebra equipped with awareness, knowledge, and belief operators satisfying some plausible axioms is representable as the algebra of events arising from a partial-state space with awareness, knowledge, and belief correspondences.\end{abstract}



\newpage

\section{Introduction}\label{Intro}

In recent decades, models of uncertainty in economics have been enriched so as to also represent \textit{unawareness}. If an agent is uncertain about an event or proposition $E$, then she can conceive of $E$ but does not know whether $E$ obtains. By contrast, if an agent is unaware of $E$, then $E$ is not even ``present to mind'' (\citealt[p.~107]{Modica1994}, \citealt[p.~274]{Modica1999}); the agent has a ``lack of conception''  of $E$ (\citealt[p.~90]{Heifetz2006}, \citealt{Schipper2015}). While unawareness is thereby distinguished from uncertainty, one prominent tradition in the literature, initiated by Modica and Rustichini \citeyearpar{Modica1994} and adopted by Heifetz, Meier, and Schipper~\citeyearpar{Heifetz2006} and Li \citeyearpar{Li2009}, takes unawareness to be definable in terms of knowledge: an agent is unaware of $E$ if and only if she does not know $E$ but also does not know that she does not know $E$. Dually, an agent is aware of $E$ if and only if she knows $E$ or knows that she does not know $E$.

While conceptually parsimonious, taking unawareness to be definable in terms of knowledge as above limits the relevant phenomena that can be modeled. Say that an agent is \textit{overconfident} when he not only believes $E$ but also believes that he knows~$E$, though he does not know $E$, say because $E$ is false.\footnote{As in \citealt{Modica1994,Modica1999}, \citealt{Heifetz2006}, and  \citealt{Li2009} (and the classics  \citealt{Aumann1999,Aumann1999b}, \citealt{Fagin1995}, and \citealt{Hintikka2005}), knowing $E$ requires that $E$ is true. As for belief, for the rest of the paragraph in the main text one may replace belief with \textit{$p$-belief} (\citealt{Monderer1989}), i.e., subjective probability of at least $p$, for any $p\in (0,1]$ for which one agrees that $p$-believing $E$ is inconsistent with knowing $\neg E$ (not $E$). For example, in Aumann's \citeyearpar[(12.2)]{Aumann1999b} framework, knowing $\neg E$ implies $1$-believing $\neg E$, so we can use any $p>0$.} For example, a potential investor in a firm might believe he knows that the firm is profitable, while in fact it is unprofitable and is committing fraud to suggest otherwise. But according to the Modica-Rustichini definition of unawareness, it is impossible for an agent to be aware of $E$ and overconfident with respect to $E$. Since the overconfident agent believes that he knows $E$, he does not know that he does not know $E$.\footnote{This simply applies the principle that if an agent believes $F$, then he does not know $\neg F$ (not $F$), which we assume in the next paragraph as well.}  Then since he also does not know $E$, according to Modica-Rustichini the overconfident agent must be unaware of $E$. But on the intuitive interpretation of awareness as being ``present to  mind'' or of unawareness as ``lack of conception,'' the overconfident agent is perfectly aware of $E$. The investor is perfectly aware of the idea of profitability of the firm; indeed, he believes (falsely) that the firm is profitable. All the behavioral predictions implied by such awareness and belief apply to the investor. Thus, if we wish to model overconfident agents, we cannot accept the Modica-Rustichini definition of unawareness.\footnote{In Appendix \ref{PlausibilityAppendix}, we argue that changing the Modica-Rustichini definition to define unawareness in terms of belief instead of knowledge still does not provide a satisfactory definition of unawareness of events.}

Moreover, under notions of belief as subjective certainty according to which believing $E$ entails believing that you know $E$ (as in \citealt[\S~3]{Stalnaker2006}), the Modica-Rustichini definition makes it impossible for an agent to be aware of $E$ and falsely believe $E$; then assuming belief requires awareness (as in \citealt{Fagin1988}), their definition makes it impossible for an agent to falsely believe $E$.

Yet we should not give up on modeling unawareness in the face of false or overconfident beliefs. The investor's overconfidence may be due in part to his unawareness of the possibility of a sophisticated type of fraud by the firm; one might predict that if he were made aware---or if we consider a different investor who is aware---he would realize that none of his due diligence ruled out such fraud.  Thus, while overconfidence with respect to $E$ should be compatible with awareness of $E$, it may be related to unawareness of some other event $F$. We would like to have a model of unawareness that can capture these~phenomena.

In this paper, we propose a new model of unawareness. Instead of defining awareness in terms of knowledge, we model awareness in a way that remains close to the paradigm of Aumann's \citeyearpar{Aumann1999} model for knowledge: just as Aumann uses a correspondence on a state space to define an agent's knowledge operator on events, we use a correspondence on a state space to define an agent's awareness operator on events.  Like the model of Heifetz, Meier, and Schipper~\citeyearpar{Heifetz2006}, our model is based on the idea of \textit{partial specifications} of the world, partially ordered by a relation of \textit{further specification} or \textit{refinement}, and the idea that agents may be aware of some coarser-grained specifications while unaware of some finer-grained specifications. However, our model is based on a different implementation of this idea with a long history in mathematics and mathematical logic. In particular, as exploited in forcing in set theory (e.g., \citealt{Takeuti1973}), partial specifications ordered by refinement  give rise to a Boolean algebra of events via the regular open sets in the downset topology of the partial order (see Section~\ref{Prelim}); crucially, in this Boolean algebra, the negation operation $\neg$ is not set-theoretic complementation, so there may be partial specifications that belong to neither $E$ nor $\neg E$.\footnote{This is analogous to rejecting the \textit{real states} assumption of \citealt{Dekel1998}, except we consider only events and not formulas in a logical language.} Next, as exploited in so-called possibility semantics for modal logic (\citealt{Humberstone1981,Holliday2014,Holliday2015,Benthem2017,HollidayForthB}), relating these partial specifications via accessibility relations---or equivalently, possibility correspondences---for different agents provides a model of multi-agent knowledge and belief, which generalizes the Kripke frames (\citealt{Kripke1963}) or Aumann structures (\citealt{Aumann1976,Aumann1999}) that drop out as special cases when using only discrete partial orders. The final step we take  here is to add a correspondence representing an agent's awareness of partial possibilities: $\nu\in \mathcal{A} (\omega)$ will mean that in possibility $\omega$, the agent is aware of possibility $\nu$.

Just as in Aumann structures one uses possibility correspondences to define knowledge \textit{of events}, in our structures we will use  $\mathcal{A} $ to define awareness of events, where an event is understood as a set of possibilities in a Boolean algebra of events. This project of modeling awareness of events in some Boolean algebra must be distinguished from the project of modeling awareness of \textit{sentences} in some language. Awareness of sentences may be \textit{hyperintensional} in the sense that where $\llbracket \varphi\rrbracket$ is the set of possibilities in which a sentence $\varphi$ is true in a given model, we may have $\llbracket \varphi_1\rrbracket=\llbracket \varphi_2\rrbracket$ while the agent is aware of the sentence $\varphi_1$ and yet unaware of the sentence $\varphi_2$.\footnote{The denial of hyperintensionality for awareness and knowledge of sentences is what Dekel et al.~\citeyearpar{Dekel1998} call \textit{event sufficiency}.} Hyperintensional models of awareness of sentences have been developed in the literature (e.g.,~\citealt{Fagin1988}, \citealt{Modica1999}, \citealt{Halpern2001}), but  here we follow the event-based tradition of Aumann \citeyearpar{Aumann1976,Aumann1999} with a non-hyperintensional model of awareness of events (see \citealt{Schipper2015} for comparison of event-based and sentence-based approaches). However, we believe these two projects should be related by the following bridge principle: an agent is aware of an event $E$ if and only if she is aware of some sentence $\varphi$ that she understands such that $\llbracket \varphi\rrbracket = E$.\footnote{This bridge principle has consequences for operationalizing the concept of awareness of events. Assuming one has a decision procedure for testing awareness of sentences, one obtains a semi-decision procedure for testing awareness of a given event $E$: enumerate sentences $\varphi_1,\varphi_2,\dots$ that express $E$ (assuming that for the given $E$, there are at most countably infinitely many sentences in the agent's language that express $E$) and check for each sentence the agent's awareness of that sentence.}

Distinguishing awareness of events vs.~sentences is crucial for assessing axioms concerning awareness. For example, awareness of events should not be monotonic with respect to the entailment relation $\leq$ in the Boolean algebra of events (which is just the subset relation $\subseteq$ when events are sets of possibilities); that is, we should not require that if $E\leq F$ ($E$~entails~$F$) and the agent is aware of $E$, then she must be aware of $F$. For example, an agent may be aware of the event $E$ expressed by `Ann and Bob will play a Nash equilibrium' without being aware of the event $F$ expressed by `Ann and Bob will play a correlated equilibrium', due to not having the concept of correlated equilibrium, despite the fact that $E\leq F$. Since $E\leq F$ is equivalent to $E = E\sqcap F$, where $\sqcap$ is the meet operation in the Boolean algebra (corresponding to intersection of sets), monotonicity is equivalent to the principle that if an agent is aware of $E\sqcap F$, then she is aware of $F$. A spurious argument for this principle is that ``if an agent is aware of a conjunction, then she must be aware of each conjunct.'' The argument is spurious because it implicitly assumes awareness of a \textit{sentence}. It is indeed plausible that if an agent is aware of a sentence $\varphi\wedge\psi$, where $\wedge$ is sentential conjunction, then the agent must be aware of $\psi$. But the event $E\sqcap F$ is not intrinsically a conjunction. We may have $E=E\sqcap F =C\sqcup D$, etc. Thus, in contrast with some other models of awareness in the literature (e.g., \citealt{Heifetz2006}, \citealt{Li2009}), our model will not validate the axiom that awareness of $E\sqcap F$ implies awareness of $F$.

Our definition of awareness of events is informally as follows. In a possibility $\omega$, an agent $i$ is aware of an event $E$ if the following condition holds at $\omega$ and its refinements: if $i$ is aware of a possibility $\nu$, then $i$ is aware of any coarsest refinement of $\nu$ belonging to $E$ and any coarsest refinement of $\nu$ belonging to $\neg E$. In other words, if you are aware of $E$, then you should be able to apply the $E$ vs.~$ \neg E$ distinction starting from any possibility of which you are aware. For example, if an agent is aware of the event $E$ expressed by `The Centers for Medicare and Medicaid Services are investigating the firm', then for any possibility $\nu$ of which the agent is aware, if $\nu$ does not already belong to $E$ or $\neg E$, then the agent should be aware of a coarsest further specification of $\nu$ belonging to $E$ and a coarsest further specification of $\nu$ belong to $\neg E$. This model of awareness  satisfies the \textit{symmetry axiom} of \citealt{Modica1994}, stating that an agent is aware of $E$ if and only if she is aware of $\neg E$ (though again without accepting the Modica-Rutschini definition of awareness in terms of knowledge), which helps distinguish being unaware of an event from assigning it zero probability (cf.~\citealt[p.~727]{Schipper2013}). Given symmetry, one may also think of the model as a model of awareness of \textit{distinctions} in the space of possibilities or of awareness of  \textit{binary questions}.

Besides symmetry, another consequence of our model is that each agent $i$ is aware of the trivial event $\Omega$ and the trivial distinction of $\Omega$ vs.~$\varnothing$. This does not imply that $i$ is aware of \textit{each possibility in} $\Omega$. Nor does it imply that $i$ is aware of each \textit{sentence} true throughout $\Omega$; e.g., it does not imply that $i$ is aware of the sentence `every Nash equilibrium is a correlated equilibrium'. This point is related to Stalnaker's \citeyearpar[pp.~85-6]{Stalnaker1984} defense of the fact that in standard state-space models of knowledge, such as Aumann structures, each agent \textit{knows} the trivial event $\Omega$: this does not imply that $i$ knows that the sentence `every Nash equilibrium is a correlated equilibrium' is true, since $i$ may fail to know the metalinguistic fact that that sentence is true throughout $\Omega$. Returning to awareness, the bridge principle proposed above implies that $i$ is aware of $\Omega$ if and only if $i$ is aware of some sentence $\varphi$ she understands such that $\llbracket \varphi\rrbracket=\Omega$, e.g., a sentence $\varphi$ such as `It is raining or it is not raining'. We therefore call the axiom that $i$ is aware of $\Omega$ the \textit{tautology axiom}.

A third consequence of our model is  the \textit{agglomeration axiom}: if $i$ is aware of events $E$ and $F$, then $i$ is also aware of $E\sqcap F$, or set theoretically, $E\cap F$. The corresponding assumption on sentential awareness is that if $i$ is aware of $\varphi$ and of $\psi$, then $i$ is aware of $\varphi\wedge\psi$. This assumes some logical sophistication on the part of $i$. Indeed, we assume that while our agents may have unawareness, they are logically perfect within the domain of their awareness. As Schipper \citeyearpar[p.~78]{Schipper2015} puts it, ``Despite such a lack of conception [i.e., unawareness], agents in economics are still assumed to be fully rational,'' in contrast to some models in computer science  of agents who are both unaware and logically imperfect (e.g., \citealt{Fagin1988}).

After handling awareness, we add knowledge and belief to our model, as in possibility semantics for modal logic. This part of our model is analogous to the standard treatment of knowledge and belief in Kripke frames and Aumann structures, allowing adjustments to deal with the partiality of possibilities. We call the resulting structures \textit{epistemic possibility frames}. We use these frames to show that an influential impossibility result concerning awareness due to Dekel, Lipman, and Rustichini  \citeyearpar{Dekel1998} becomes a possibility result (Fact \ref{PossibilityFact}) under a slight weakening of one of their axioms.

We can now informally state the main technical result about our model, which is a representation theorem. We define an \textit{epistemic awareness algebra} to be a Boolean algebra equipped with an awareness operator satisfying the axioms of symmetry, tautology, and agglomeration, plus knowledge and belief operators satisfying some minimal axioms. We then prove (Theorem \ref{FilterRep}) that any epistemic awareness algebra is representable as the algebra of events of an epistemic possibility frame with the awareness operator represented using the awareness correspondence $\mathcal{A} $ as sketched above and with knowledge and belief operators represented using knowledge and belief correspondences. This theorem shows that our model of awareness given by a partially ordered set equipped with regular open events and the correspondences for awareness, knowledge, and belief is highly versatile: it can represent any situation involving agents' event-based awareness, knowledge, and belief, provided some basic axioms are satisfied. 

There is now a large literature on unawareness in economics and computer science, as surveyed up to around 2014 in \citealt{Schipper2014,Schipper2015}. A non-exhaustive list of more recent contributions in economics includes \citealt{Grant2015}, \citealt{Quiggin2016}, \citealt{Karni2017},  \citealt{Galanis2016,Galanis2018}, \citealt{Piermont2017}, \citealt{Dietrich2018}, \citealt{Guarino2020}, \citealt{Fukuda2021},  \citealt{Schipper2021b}, \citealt{Dominiak2022}, and the special issue introduced by \citealt{Schipper2021}. Here we will briefly discuss  those works most relevant to the present paper.

\subsection{Related models}

The model of unawareness that is most often used in applications and is related in spirit to the model of this paper is that of Heifetz, Meier, and Schipper \citeyearpar{Heifetz2006} mentioned above, known as the HMS model (see \citealt{Heifetz2008}, \citealt{Halpern2008}, and \citealt{Belardinelli2020} on the relation between this model and syntactic, logical approaches).  As noted, the mathematical implementation of the idea of partial states in our model differs from that in the HMS model.  Other differences include the following: (i)~The algebra of events in our model forms a Boolean algebra, whereas the algebra of ``events'' in the HMS model is not even a lattice under their operations for conjunction and disjunction (see Remark \ref{AlgRemark});\footnote{\label{SubjectMatter}One way to make sense of this fact about the HMS model is to think that what HMS call ``events'' are not events in the ordinary sense in decision theory, or what philosophers call (coarse-grained) \textit{propositions}, which are determined entirely by their truth conditions. Instead, they are hyperintensional entities that have not only a truth-conditional component but also a non-truth-conditional component consisting of, e.g., subject matter, or ``expressive power'' (\citealt[p.~80]{Heifetz2006}), etc.; indeed, the poset of ``events'' in the HMS model, as described in Appendix~A of \citealt{Heifetz2008}, is isomorphic to a poset of  pairs $(P,S)$ where $P$, the truth-conditional part, is a subset of a maximally rich state space, and $S$, the part encoding subject matter or expressive power, is a possibly less rich state space, such that $P$ is the inverse image of a subset of $S$ with respect to the  projection of the maximally rich state space onto $S$ (the partial order is then given by $(P,S)\leq (P',S')$ if $P\subseteq P'$ and $S'\preceq S$, where $\preceq$ is HMS's complete lattice order of state spaces). Philosophers have called a pair of a proposition and a subject matter a \textit{directed proposition} (\citealt[p.~49]{Yablo2014}). For these entities, it is easy to see why a lattice axiom like $E=E\sqcap (E\sqcup F)$ (the absorption law) can fail, as it fails in the HMS model (see Remark \ref{AlgRemark}), because the directed proposition $F$ can introduce new subject matter beyond that of $E$, whereas for purely truth-conditional propositions or events, the axiom holds.} (ii)~We model awareness using an Aumann-style possibility correspondence, whereas HMS define awareness in terms of knowledge as in the Modica-Rustichini definition; and (iii)~Our model avoids the impossibility theorem of Dekel, Lipman, and Rustichini  \citeyearpar{Dekel1998} by falsifying a different axiom than the HMS model falsifies (see Appendix~\ref{DLR}). In light of points (i)--(ii), in a way our model constitutes a more ``conservative'' approach to modeling unawareness than the HMS model; we see how far we can go in modeling unawareness using  Boolean algebras of events and ideas from classical modal logic, for which so much theory has already been developed.

Another closely related model of unawareness is that of Fritz and Lederman \citeyearpar{Fritz2015}.  They provide a model of awareness of events that also validates the three axioms of symmetry, tautology, and agglomeration, while also rejecting the principle that awareness of $E\sqcap F$ implies awareness of $F$ (cf.~Theorem 3 in their paper) and rejecting an axiom of Dekel et al.~\citeyearpar{Dekel1998}  that we also reject (namely the Plausibility axiom, discussed in our Appendix \ref{PlausibilityAppendix}). They do so via a very different construction than ours,\footnote{\label{FritzLederman}In particular, they introduce structures $(\Omega, \approx  )$ where $\Omega$ is a nonempty set of states and $\approx $ assigns to each $\omega\in\Omega$ an equivalence relation $\approx_\omega$ on $\Omega$; then the agent is aware of $E\subseteq\Omega$ in state $\omega$, denoted $\omega\in a(E)$, if and only if for all $\rho$ and $\tau$ such that $\rho\approx_\omega \tau$, we have $\rho\in E$ if and only if $\tau\in E$. This approach and ours locate the complexity of modeling awareness in different places. Theirs is more complex in assigning to each state $\omega$ an equivalence relation $\approx_\omega$ on $\Omega$, whereas ours simply assigns a subset $\mathcal{A} (\omega)\subseteq\Omega$; but ours is more complex in working with a partially ordered set $(\Omega,\sqsubseteq)$, whereas theirs simply works with a set $\Omega$.} which we take to be further evidence of the naturality of the three axioms.  There is also a fundamental difference in representational power. On their approach, the family of all $E\subseteq\Omega$ of which an agent is aware at a state must be an \textit{atomic} algebra of sets.\footnote{\label{AtomNote}The atoms of this algebra are the equivalence classes of the relation $\approx_\omega$ from Footnote \ref{FritzLederman}. Recall that an \textit{atom} in a Boolean algebra $(B,\leq)$  (see Section \ref{Prelim} for the definition of Boolean algebras as special partially ordered sets) is an $a\in B$ such that $0<a$ and there is no $b$ with $0<b<a$. A Boolean algebra is \textit{atomic} if for each $b\in B$, there is an atom $a\leq b$. By contrast, it is \textit{atomless} if it has no atoms. An algebra $\mathcal{E}$ of sets is atomic (resp.~atomless) if it is atomic (resp.~atomless) when regarded as a Boolean algebra $(\mathcal{E},\subseteq)$.} By contrast, on our approach, even if our ambient set of events is an atomic Boolean algebra (which it need not be), there is no requirement that the Boolean subalgebra of events of which the agent is aware is atomic (see Example \ref{InfiniteEx}). Finally, Fritz and Lederman \citeyearpar[Appendix B]{Fritz2015}  prove the weak completeness,\footnote{As usual in logic (see, e.g., \citealt[p.~194]{Blackburn2001}), a \textit{weak completeness} theorem states that for every formula $\varphi$, if $\varphi$ is semantically valid, then $\varphi$ is syntactically provable; a \textit{strong completeness} theorem states that for every set $\Gamma$ of formulas and formula $\varphi$, if $\varphi$ is a semantic consequence of $\Gamma$, then $\varphi$ is syntactically provable from assumptions in $\Gamma$.} with respect to a semantics based on their structures, of a modal logic of awareness with axioms of symmetry, tautology, and agglomeration. Although we do not introduce logical syntax in this paper, our representation theorem for arbitrary epistemic awareness algebras (Theorem \ref{FilterRep}) rather immediately yields strong completeness theorems for the logic of awareness---and extensions thereof---but now with respect to a semantics based on our epistemic possibility frames.

Finally, as far as I know, the only other work besides the present paper that considers both awareness in economics and possibility semantics from modal logic is \citealt{Piermont2024}.  Piermont relates both to his concept of a \textit{relativized Boolean algebra}, which is an algebraic structure satisfying some of the laws of Boolean algebras but not $E\sqcup \neg E=1$.\footnote{A standard construction (see, e.g., \citealt[pp.~91-2]{Givant2009}) associates with a Boolean algebra $\mathbb{B}$ and element $a$ in $\mathbb{B}$ a new Boolean algebra $\mathbb{B}(a)$, the \textit{relativization of $\mathbb{B}$ to $a$}, whose bottom element and meet operation coincide with those of $\mathbb{B}$ but whose top element $1_a$ is $a$ and whose complement and join operations are defined by $\neg_a c = \neg c\sqcap a$ and $c\sqcup_a d=(c\sqcup d)\sqcap a$. Examples of relativized Boolean algebras can be obtained by lifting this construction to  $\mathbb{B}_\star=\{(a,b)\mid a,b\in\mathbb{B}, a\leq b\}$ as follows:  $\neg_\star (a,b)=(\neg_b a, b)$, $(a,b)\sqcap_\star (a',b') =(a\sqcap a', b\sqcap b')$, and $(a,b)\sqcup_\star (a',b')=(a\sqcup_{b\sqcap b'}a', b\sqcap b')$, $0_\star =(0,0)$, and $1_\star = (1,1)$.} By contrast, we work with the Boolean algebra of regular open sets canonically associated with a partially ordered set (Theorem \ref{Tarski2}); and we explain away intuitions that due to unawareness  $E\sqcup\neg E$ may fail to equal $\Omega$ by insisting on the distinction between events and sentences (see Example \ref{WatsonExample}) or other hyperintensional entities (see Footnote \ref{SubjectMatter}). Our classical approach makes available all standard results, e.g., from measure theory, that are applicable to Boolean algebras.

\subsection{Organization}\label{Organization}

The rest of the paper is organized as follows. In Section~\ref{Prelim}, we review mathematical preliminaries that underlie our model. In Section~\ref{ModelSection}, we introduce the model of awareness (\ref{AwarenessSection}) and then add knowledge and belief (\ref{KnowledgeSection}), using the model to formalize several examples. In Section~\ref{Representation}, we state our main representation theorem.  In Section~\ref{Conclusion}, we conclude with directions for future work, including a sketch of how to add awareness of sentences and probability to our framework. All substantial proofs are collected in Appendix \ref{Proofs}. 

Appendix~\ref{DLR} reviews the impossibility theorem of Dekel et al.~\citeyearpar{Dekel1998}, which threatens to preclude the development of a model of awareness of events as opposed to sentences. We argue that one of their axioms is too strong, and this is the axiom whose weakening yields a possibility result in Section \ref{KnowledgeSection}.

A Jupyter notebook containing code to verify conditions and compute awareness, knowledge, and belief in all examples in this paper is available at  \href{https://github.com/wesholliday/awareness}{https://github.com/wesholliday/awareness}.

\section{Preliminaries}\label{Prelim}

Standard representations of uncertainty begin with a nonempty set $\Omega$ of \textit{states of the world}, whose powerset we denote by  $\wp(\Omega)$. The set of \textit{events} is then some nonempty  collection $\mathcal{E}\subseteq \wp(\Omega)$ closed under at least finite intersection and complement relative to $\Omega$. Here we instead begin with a partially ordered set (poset) $(\Omega,\sqsubseteq)$. We think of elements of $\Omega$ as \textit{partial possibilities} and take $\omega\sqsubseteq \nu$ to mean that $\omega$ is a \textit{further specification} or \textit{refinement} of $\nu$. For example, the possibility $\nu$ may settle that \textit{Ann plays up} in a game but leave undetermined what Bob plays; then a refinement $\omega$ of $\nu$ may settle not only that \textit{Ann plays up} but also that \textit{Bob plays left}, while another refinement $\omega'$ of $\nu$ may settle that \textit{Ann plays up} and \textit{Bob plays right}. With this picture, not every subset of $\Omega$ is eligible to count as an event. We will delimit the eligible events~shortly.

Given a poset $(\Omega,\sqsubseteq)$, we define the \textit{downward closure} operation $\mathord{\downarrow}$ on $\wp(\Omega)$ by
\[\mathord{\downarrow}E=\{\omega\in \Omega\mid \mbox{for some }\nu\in E, \omega\sqsubseteq \nu\}.\]
A set $E\subseteq \Omega$ is a \textit{downset of $(\Omega,\sqsubseteq)$} if $E=\mathord{\downarrow}E$. For $\omega\in\Omega$, we write $\mathord{\downarrow}\omega$ for $\mathord{\downarrow}\{\omega\}$, called a \textit{principal downset}. Possibilities $\omega$ and $\nu$ are \textit{compatible}, denoted $\omega\between\nu$, if  $\mathord{\downarrow}\omega\cap\mathord{\downarrow}\nu\neq\varnothing$; otherwise they are \textit{incompatible}, denoted $\omega \, \bot \,\nu$. The \textit{downset topology} on $\Omega$ is the topology whose open sets are exactly the downsets of $(\Omega,\sqsubseteq)$. The interior and closure operations for this topology are given by
\begin{eqnarray*}\mathsf{int}(E)&=&\{\omega\in \Omega\mid \mbox{for all }\nu\sqsubseteq \omega,\, \nu\in E\}\\
\mathsf{cl}(E)&=&\{\omega\in \Omega\mid \mbox{for some }\nu\sqsubseteq \omega,\, \nu\in E\}.
\end{eqnarray*}
The  \textit{regularization} operation $\rho:\wp(\Omega)\to\wp(\Omega)$ is then defined by \[\rho(E)= \mathsf{int}(\mathsf{cl}(\mathord{\downarrow}E))=\{\omega\in \Omega\mid \forall \omega'\sqsubseteq \omega\;\,\exists \omega''\sqsubseteq \omega' \colon \omega''\in \mathord{\downarrow}E\}.\]
A set $E\subseteq \Omega$ is \textit{regular open} if $\rho(E)=E$. Let $\mathcal{RO}(\Omega,\sqsubseteq)$ be the collection of all regular open sets. Note that if $\sqsubseteq$ is the discrete partial order (i.e., the identity relation) on $\Omega$, then $\mathcal{RO}(\Omega,\sqsubseteq)$ is just $\wp(\Omega)$.

Regular open sets can be characterized by the following conditions. The proof is straightforward.

\begin{lemma}\label{ROLem} Given a poset $(\Omega,\sqsubseteq)$ and $E\subseteq \Omega$, we have $E\in\mathcal{RO}(\Omega,\sqsubseteq)$ if and only if for all $\omega,\omega'\in\Omega$:
\begin{enumerate}
\item \textit{persistence}: if $\omega\in E$ and $\omega'\sqsubseteq \omega$, then $\omega'\in E$;
\item \textit{refinability}: if $\omega\not\in E$, then $\exists \omega'\sqsubseteq \omega$ $\forall \omega''\sqsubseteq \omega'$ $\omega''\not\in E$.
\end{enumerate}
\end{lemma}
\noindent Only regular open sets will count as genuine \textit{events} in our model. Persistence states that if a possibility $\omega$ settles that an event holds, so does any refinement of $\omega$. Refinability states that if a possibility  $\omega$ does not settle that an event holds, then there is a refinement $\omega'$ of $\omega$ that settles that the event does \textit{not} hold, in the sense that no possible refinement of $\omega'$ settles that the event holds (cf.~the definition of $\neg$ in Theorem \ref{Tarski2}).

A poset is \textit{separative} if for any $\omega\in \Omega$, its principal downset $\mathord{\downarrow}\omega$ is a regular open set, which may be described as the following event: \textit{the possibility $\omega$ obtains}. One may assume without loss of generality in what follows that all posets are separative.\footnote{\label{Quotient}One can always pass to a quotient poset by identifying $\omega$ and $\omega'$ when $\rho(\{\omega\})=\rho(\{\omega'\})$ and defining the relation on equivalence classes by $[\omega]\sqsubseteq [\nu]$ if $\omega\in \rho(\{\nu\})$, resulting in a separative poset with an isomorphic algebra of regular open sets.} An example of a non-separative poset is a two-element linear order with $\omega\sqsubseteq \nu$; the principal downset $\{\omega\}$ does not satisfy refinability, since $\nu\not\in \{\omega\}$ and yet there is no refinement of $\nu$ all of whose refinements are not in $\{\omega\}$. Collapsing this two-element linear order to a linear order on the singleton $\{\omega\}$ results in a separative poset with an isomorphic algebra of regular open sets (as in Theorem~\ref{Tarski2}), consisting of the empty set and the whole set. Another way to get a separative poset from the two-element linear order---but in a way that results in a non-isomorphic algebra of regular opens---is to add a third element $\omega'$ with $\omega'\sqsubseteq \nu$, $\omega'\not\sqsubseteq \omega$, and $\omega\not\sqsubseteq \omega'$, so we obtain the three-element tree; now $\{\omega\}$ satisfies refinability, since $\nu$ is refined by $\omega'$, all of whose refinements (namely just $\omega'$ itself) are not in $\{\omega\}$.

 Given an arbitrary set $E$ of possibilities, $\rho(E)$ is the event that may be described as: \textit{one of the possibilities in $E$ obtains}. In the language of lattice theory, $\rho$ is not only a \textit{closure operator} (satisfying conditions 1-3 in Lemma~\ref{RhoLem}) on $\wp(\Omega)$ but also a \textit{nucleus} (satisfying 4 in addition to 1-3) on downsets. The proof is straightforward.

\begin{lemma}\label{RhoLem} For any poset $(\Omega,\sqsubseteq)$,  $\rho$ satisfies the following for all $E,F\subseteq \Omega$: 1. if $E\subseteq F$, then $\rho(E)\subseteq \rho(F)$;  2. $E\subseteq \rho(E)$; and 3. $\rho(\rho(E))=\rho(E)$. Moreover, if $E,F$ are downsets, then 4.~$\rho(E)\cap \rho(F)\subseteq \rho(E\cap F)$.
\end{lemma}

Next we characterize the poset $(\mathcal{RO}(\Omega,\sqsubseteq),\subseteq)$. Recall that a poset ${(L,\leq)}$ is a \textit{lattice} (resp.~\textit{complete lattice}) if every two-element subset $\{a,b\}\subseteq L$ (resp.~every subset $\{a_i \}_{i\in I} \subseteq L$) has a least upper bound with respect to $\leq$, called the \textit{join} and denoted $a\sqcup b$ (resp.~$\underset{i\in I}{\bigsqcup} a_i $) and greatest lower bound with respect to $\leq$, called the \textit{meet} and denoted  $a\sqcap b$ (resp.~$\underset{i\in I}{\bigsqcap} a_i $). A lattice is \textit{distributive} if for all $a,b,c\in L$, we have $a\sqcap (b\sqcup c)= (a\sqcap b)\sqcup (a\sqcap c)$. A lattice is \textit{bounded} if it has a greatest element with respect to $\leq$, denoted $1$, and a least element with respect to $\leq$, denoted $0$.\footnote{Every complete lattice is bounded, since the least upper bound of $\varnothing$ is 0 and the greatest lower bound of $\varnothing$ is 1.}  A bounded lattice is \textit{complemented} if for every $a\in L$, there is an  $\neg a\in L$, called a \textit{complement of $a$}, such that $a\sqcup\neg a=1$ and $a\sqcap\neg a =0$ (the complement of $a$ is unique if the lattice is distributive).  A \textit{Boolean algebra} is a complemented distributive lattice $\mathbb{B}=(B,\leq)$. We abuse notation and write $a\in \mathbb{B}$ for $a\in B$.

It is a classic result in lattice theory that the fixpoints of any closure operator on $\wp(\Omega)$, ordered by $\subseteq$, form a complete lattice with meet as intersection and join as closure of union (\citealt[Thm.~5.2]{Burris1981}); and as Tarski \citeyearpar{Tarski1937} observed, the fixpoints of the regularization operation $\rho$ form a complete Boolean algebra.  For modern proofs of the following, see, e.g., \citealt[Thms.~1.30, 1.40]{Takeuti1973}.

\begin{theorem}\label{Tarski2} For any poset $(\Omega,\sqsubseteq)$, the poset $(\mathcal{RO}(\Omega,\sqsubseteq),\subseteq)$ is a complete Boolean algebra, called the \textit{regular open algebra of $(\Omega,\sqsubseteq)$}, in which the Boolean complement, meet, and join are given by:
\begin{eqnarray}\neg E&=&\mathsf{int}(\Omega\setminus E)=\{\omega\in\Omega\mid \forall \omega'\sqsubseteq \omega\;\, \omega'\not\in E\}\label{NegDef}\\
\underset{i\in I}{\bigsqcap}E_i &=&\underset{i\in I}{\bigcap}E_i \\
\underset{i\in I}{\bigsqcup}E_i &=&\rho\Big(\underset{i\in I}{\bigcup}E_i \Big)=\{\omega\in\Omega\mid \forall \omega'\sqsubseteq \omega\;\,\exists \omega''\sqsubseteq \omega'\;\,\exists i\in I\colon \omega''\in E_i \}.\label{JoinDef}
\end{eqnarray}
Conversely, each complete Boolean algebra $(B,\leq)$ is isomorphic to $\mathcal{RO}(B_+,\leq_+)$, where $B_+$ is the set of nonzero elements of the algebra and $\leq_+$ is $\leq$ restricted to $B_+$, via the map $b\mapsto \{a\in B_+\mid a\leq b\}$.
\end{theorem}

\begin{remark}\label{AlgRemark} In the regular open algebra of a poset, we always have $ \underset{i\in I}{\bigcup}E_i \subseteq \underset{i\in I}{\bigsqcup}E_i $ and often $ \underset{i\in I}{\bigcup}E_i \subsetneq \underset{i\in I}{\bigsqcup}E_i $. This represents a deep difference between our approach and that  of Heifetz et al.~\citeyearpar{Heifetz2006}, who define their disjunction of events in such a way that often $E_j \not\leq \underset{i\in I}{\bigsqcup}E_i $ for $j\in I$, where $\leq$ is their partial order on events (see Appendix~\ref{DLR}). Since their $E\sqcup F$ is not necessarily an upper bound of $\{E,F\}$ with respect to $\leq$, their algebra of events with $\sqcup$ and $\sqcap$ is not even a lattice,\footnote{In equational terms, it violates the absorption law of lattices that $E\sqcap (E\sqcup F) = E$ (see Footnote \ref{SubjectMatter} for discussion).} let alone a Boolean algebra. In addition, their negation operation $\neg$ is non-classical, as it violates the equivalence of $E\leq F$ and $\neg F\leq\neg E$ from Boolean algebras.
\end{remark}

In applications to reasoning under uncertainty (e.g., involving probability) a Boolean algebra of events is often not complete (though it may be countably complete), so we do not want to restrict attention to only representing complete Boolean algebras of events. To represent arbitrary Boolean algebras, we can equip a poset $(\Omega,\sqsubseteq)$ with a distinguished Boolean subalgebra  of $\mathcal{RO}(\Omega,\sqsubseteq)$.

\begin{definition}\label{PossFrames} A \textit{possibility frame} is a triple $(\Omega,\sqsubseteq,\mathcal{E})$ where $(\Omega,\sqsubseteq)$ is a  poset and $\mathcal{E}$ is a nonempty subset of $\mathcal{RO}(\Omega,\sqsubseteq)$ closed under binary intersection and the operation $\neg$ from~(\ref{NegDef}).
\end{definition}

Compare the notion of a possibility frame to the more familiar notion of a \textit{field of sets}, a pair $(\Omega,\mathcal{E})$ where $\Omega$ is a nonempty set and $\mathcal{E}$ is an algebra of subsets of $\Omega$ as at the beginning of this section. It is a classic result of Stone \citeyearpar{Stone1936} that for each Boolean algebra $\mathbb{B}$, there is a field of sets $(\Omega,\mathcal{E})$ such that $\mathbb{B}$ is isomorphic to $(\mathcal{E},\subseteq)$. Since every field of sets may be regarded as a possibility frame $(\Omega,\sqsubseteq,\mathcal{E})$ in which $\sqsubseteq$ is the discrete partial order, Stone's theorem immediately implies that each Boolean algebra is representable by a possibility frame. But we will need non-discrete partial orders to model unawareness using possibility frames. As a consequence of Theorem \ref{FilterRep}, we will obtain for each Boolean algebra $\mathbb{B}$ a possibility frame $(\Omega,\sqsubseteq,\mathcal{E})$ with a non-discrete partial order such that $\mathbb{B}$ is isomorphic to~$(\mathcal{E},\subseteq)$.

 Given a poset $(\Omega,\sqsubseteq)$ and $E\in \mathcal{RO}(\Omega,\sqsubseteq)$, we denote the set of maximal elements of $E$ by
\[\mathrm{max}(E)=\{\omega\in E\mid \mbox{there is no }\nu\in E\colon \omega\sqsubset \nu\},\]
where $\omega\sqsubset \nu$ means that $\omega\sqsubseteq \nu$ and $\nu \not\sqsubseteq \omega$. Intuitively, $\mathrm{max}(E)$ contains the \textit{coarsest} or \textit{least refined} possibilities that settle that $E$ holds. It is a natural thought that for any nonempty event $E$, there should be a \textit{unique} coarsest possibility belonging to $E$, describable as ``the possibility that $E$ holds''; this is indeed the case for the possibility frame $(B_+,\leq_+,\mathcal{RO}(\Omega,\sqsubseteq))$ used in Theorem \ref{Tarski2}, and in fact we shall see that every Boolean algebra (not only complete ones) can be represented by a possibility frame satisfying this condition (Theorem \ref{FilterRep}). However, since in applications we typically wish to draw as few possibilities as possible to model a given situation,\footnote{Cf.~Examples \ref{GameEx} and \ref{OverconfidentEx}. In the former, there is no coarsest possibility in the event $\underline{Middle}$ (resp.~$\underline{Up}$, $\underline{Down}$), and in the latter, there is no coarsest possibility in the event $Fraud$.} we impose only the following less demanding condition.

\begin{definition}\label{QP} A possibility frame $(\Omega,\sqsubseteq,\mathcal{E})$ is \textit{quasi-principal} if for any $E\in \mathcal{E}$ and $\omega\in E$, we have $\omega\in \mathord{\downarrow}\mathrm{max}(E)$.
\end{definition}
\noindent In other words, any possibility that settles that $E$ holds is a refinement of some coarsest possibility that settles that $E$ holds. Of course any possibility frame with $\Omega$ finite satisfies this condition.

 \section{Model}\label{ModelSection}
 
 In this section, we introduce our model in two stages: first concentrating on awareness in Section \ref{AwarenessSection} and then adding knowledge and belief alongside awareness in Section \ref{KnowledgeSection}.
 
 Just as the basic datum in Aumann's \citeyearpar{Aumann1999} model of knowledge is a correspondence $\mathcal{K} :\Omega \to \wp(\Omega)$, the basic datum in our model of awareness is a correspondence $\mathcal{A} :\Omega\to\wp(\Omega)$. Intuitively, $\omega'\in \mathcal{K} (\omega)$ means that in state $\omega$, the agent's knowledge does not rule out state $\omega'$, i.e., everything the agent knows in $\omega$ is true in $\omega'$. A standard gloss on $\omega'\in\mathcal{K}(\omega)$ is that ``in $\omega$, the agent considers $\omega'$ possible,'' but this is not a good gloss; for the agent might be totally unaware of $\omega'$, while at the same time the agent's knowledge does not rule out $\omega'$, so $\omega'\in \mathcal{K} (\omega)$.\footnote{There is another reason that ``the agent considers $\omega'$ possible'' is not a good gloss on $\omega'\in \mathcal{K}(\omega)$, which is independent of unawareness. For example, in a game, Ann might be fully confident that Bob does not hold an ace, whereas Bob does in fact hold an ace. Then Ann does not ``consider it possible'' that Bob holds an ace, and yet Ann's knowledge does not rule out that Bob holds an ace, because Ann's knowledge cannot rule out the actual state of the world.} Turning to $\mathcal{A}$, the intuitive interpretation of $\omega'\in\mathcal{A} (\omega)$ is that in $\omega$, the agent is aware of $\omega'$ as a logical possibility. This does not mean that the agent ``considers $\omega'$ possible'' in the sense of thinking that $\omega'$ might  actually obtain, for the agent may be fully convinced that it does not obtain; but the agent can at least entertain $\omega'$. For example, I can entertain possibilities in which, e.g., fusion power is the leading source of global energy, even though such possibilities are not consistent with my current knowledge. Presumably some people, e.g., small children, cannot even presently entertain such possibilities.
 
Mathematically, it is important to recall what makes modeling knowledge with a correspondence possible. Suppose we begin with a knowledge operator $\mathbf{K} : \wp(\Omega)\to\wp(\Omega)$; in fact, it is convenient to start with the dual operator $\widehat{\mathbf{K}}: \wp(\Omega)\to\wp(\Omega)$ such that $\widehat{\mathbf{K}}(E)=\neg \mathbf{K} \neg (E)$, so $\omega\in  \widehat{\mathbf{K}} (E)$ means that $E$ is consistent with the agent's knowledge in $\omega$. The key idea of Aumann's model, like those of Hintikka \citeyearpar{Hintikka2005} and Kripke \citeyearpar{Kripke1963}, is to reduce $\widehat{\mathbf{K}} $ to its behavior on singleton events $\{\nu\}$, in the sense that we want:
\begin{equation}\omega\in \widehat{\mathbf{K}}(E)\mbox{ if and only if for some }\nu\in E: \omega\in \widehat{\mathbf{K}}(\{\nu\}).\label{ReduceToSingletons}\end{equation}
If this holds, and only if this holds,  we can represent the operator $\widehat{\mathbf{K}}:  \wp(\Omega)\to\wp(\Omega)$ using a simpler correspondence $\mathcal{K} :\Omega \to \wp(\Omega)$ defined by 
\begin{equation}\mbox{$\nu\in \mathcal{K} (\omega)$ if and only if $\omega\in \widehat{\mathbf{K}} (\{\nu\})$},\label{InduceKnowledgeCorr}\end{equation}
where the representation of $\widehat{\mathbf{K}} $ has the form: 
\begin{equation}\omega\in \widehat{\mathbf{K}}(E)\mbox{ if and only if for some }\nu\in E: \nu\in\mathcal{K} (\omega).\label{ReduceToCorr}\end{equation}
This is analogous to reducing a probability measure on $\wp(\Omega)$ to its values on singleton events, which is always possible in the finite case and is also possible in the countably infinite case assuming the probability measure is countably additive. Similarly, the reducibility of $\widehat{\mathbf{K}} $ to its behavior on singleton events as in (\ref{ReduceToSingletons})---and hence the representability of $\widehat{\mathbf{K}} $ using a  correspondence as in (\ref{ReduceToCorr})---is equivalent to $\widehat{\mathbf{K}} $ being \textit{completely additive}, in the sense that for any family of events $\{E_j\}_{j\in J}\subseteq \wp(\Omega)$:
\[\widehat{\mathbf{K}} (\underset{j\in J}{\bigcup} E_j)=\underset{j\in J}{\bigcup} \widehat{\mathbf{K}} (E_j).\] 

Our approach to awareness is analogous to Aumann's approach to knowledge: we will reduce the behavior of an awareness operator $\mathbf{A} :\wp(\Omega)\to\wp(\Omega)$ to its behavior on special events, but since we work with a partially ordered set rather than a set, these special events will be \textit{principal downsets} $\mathord{\downarrow}\omega$ rather than singletons $\{\omega\}$. Thus, our awareness correspondence $\mathcal{A} :\Omega\to\wp(\Omega)$ will be such that 
\begin{equation}\mbox{$\nu\in\mathcal{A} (\omega)$ if and only if $\omega\in \mathbf{A} (\mathord{\downarrow}\nu)$.}\label{InduceAwareCorr}\end{equation}
Unlike Aumann's representation of knowledge, however, the representation of $\mathbf{A} $ using $\mathcal{A} $ will not have a form like (\ref{ReduceToCorr}) above. It will have a different form, given in Definition \ref{PossFrameAware}.\ref{PossFrameAware3} below, that reflects the distinction between an event being consistent with an agent's knowledge and the agent being aware of the event.

 \subsection{Awareness}\label{AwarenessSection}
 
As discussed above, to model awareness we equip possibility frames as in Definition \ref{PossFrames} with an awareness correspondence $\mathcal{A} $. When modeling multiple agents, we can simply introduce a correspondence $\mathcal{A}_i$ for each agent~$i$, but for simplicity here we present everything in the single-agent setting. When  $\nu \in \mathcal{A} (\omega)$, we say that in possibility $\omega$, the agent is aware of possibility $\nu$.

\begin{definition}\label{PossFrameAware} A \textit{possibility frame with awareness} is a tuple $\mathscr{F}=(\Omega,\sqsubseteq, \mathcal{E} , \mathcal{A})$  such that:
\begin{enumerate}
\item\label{PossFrameAware1} $(\Omega,\sqsubseteq,\mathcal{E})$ is a quasi-principal possibility frame with a maximum element $\maximum $ in the poset $(\Omega,\sqsubseteq)$;
\item\label{AwarenessConditions} $\mathcal{A} :\Omega\to\wp(\Omega)$ is a correspondence satisfying the following  conditions for all $\omega,\omega',\nu\in\Omega$:
\begin{enumerate}
\item awareness nonvacuity: $\maximum \in \mathcal{A} (\omega)$;
\item\label{PossFrameAware2b} awareness expressibility: if $\nu \in \mathcal{A} (\omega)$, then $\mathord{\downarrow}\nu\in \mathcal{E}$;
\item awareness persistence: if $\omega'\sqsubseteq \omega$, then $\mathcal{A} (\omega)\subseteq \mathcal{A} (\omega')$;
\item awareness refinability: if $\nu \not\in \mathcal{A} (\omega)$, then $\exists \omega'\sqsubseteq \omega$ $\forall \omega''\sqsubseteq \omega'$  $\nu \not\in \mathcal{A} (\omega'')$;
\item awareness joinability: if $\nu \in \mathcal{A} (\omega)$, $E,E'\in\mathcal{E}$, and $\mathrm{max}(E\cap \mathord{\downarrow}\nu)\cup \mathrm{max}(E'\cap \mathord{\downarrow}\nu)\subseteq \mathcal{A} (\omega)$, then $\mathrm{max}((E\sqcup E')\cap\mathord{\downarrow}\nu)\subseteq \mathcal{A} (\omega)$.
\end{enumerate}
\item\label{PossFrameAware3} $\mathcal{E}$ is closed under the operation $E\mapsto \mathbf{A}  (E)$ defined by
\begin{itemize}
\item\label{EventClosure} $\omega\in \mathbf{A} (E)$ if and only if $\forall \omega'\sqsubseteq \omega$ $\forall \nu\in \mathcal{A} (\omega')$ $\mathrm{max}(E\cap\mathord{\downarrow}\nu)\cup \mathrm{max}(\neg E\cap\mathord{\downarrow}\nu)\subseteq \mathcal{A} (\omega')$.
\end{itemize}

\end{enumerate}
\noindent Finally, we call $\mathscr{F}$ \textit{standard} if for all $\omega,\nu\in\Omega$, $\nu\in\mathcal{A} (\omega)$ implies $\omega\in\mathbf{A} (\mathord{\downarrow}\nu)$.\end{definition}

The interpretations of the conditions on $\mathcal{A} $ are as follows. Awareness nonvacuity says that each agent $i$ is aware of at least the coarsest possibility of all. Awareness expressibility says that if $i$ is aware of a possibility, then the principal downset generated by that possibility is a genuine event, eligible to be thought about.\footnote{Recall from Section \ref{Prelim} that in a separative poset, every principal downset is a regular open set in the downset topology.} Awareness persistence says that if $\omega$ settles that $i$ is aware of a possibility $\nu $, and $\omega'$ refines $\omega$, then $\omega'$ still settles that $i$ is aware of $\nu $. Awareness refinability says that if $\omega$ does not settle that $i$ is aware of $\nu $, then there is a refinement $\omega'$ of $\omega$ that settles that $i$ is definitely \textit{not} aware of $\nu $, so no refinement $\omega''$ of $\omega'$ settles that $i$ is aware of $\nu $. Finally, awareness joinability says that if $i$ is aware of $\nu $ and of the coarsest refinements of $\nu$ belonging to the event $E$, and similarly for $E'$, then $i$ must be aware of the coarsest refinements of $\nu $ belonging to the event  \textit{$E$ or $E'$}. There is a convenient equivalent condition if $\Omega$ is finite, which quantifies over possibilities rather than arbitrary events:  if $i$ is aware of $\nu $ and some refinements $\nu_1,\dots,\nu_n$ of $\nu $, then $i$ must be aware of the coarsest refinements of $\nu $ belonging to the event that \textit{one of the $\nu_j $'s obtains}.

\begin{restatable}{lemma}{FiniteJoin} Suppose $\mathscr{F}=(\Omega,\sqsubseteq, \mathcal{E} , \mathcal{A})$ satisfies part \ref{PossFrameAware1} of Definition \ref{PossFrameAware} and awareness expressibility. If (i) $\mathscr{F}$ satisfies awareness joinability, then (ii) for all $\nu \in \mathcal{A} (\omega)$ and $\nu_1,\dots,\nu_n\in \mathcal{A} (\omega)\cap \mathord{\downarrow}\nu$, we have $\mathrm{max}((\mathord{\downarrow}\nu_1\sqcup\dots\sqcup\mathord{\downarrow}\nu_n)\cap\mathord{\downarrow}\nu)\subseteq \mathcal{A} (\omega)$. Conversely, if $\mathrm{max}(E)$ is finite for each $E\in\mathcal{E}$, then (ii) implies (i).
\end{restatable}

\noindent Particular frames used in applications may of course satisfy additional conditions on $\mathcal{A} $ (cf.~Remark \ref{ExtraConditions}). 

The definition of the $\mathbf{A} $ operation in part \ref{EventClosure} formalizes the account of awareness of events sketched in Section \ref{Intro}, namely that in a possibility $\omega$, an agent $i$ is aware of an event $E$ if the following condition holds at $\omega$ and its refinements: if $i$ is aware of a possibility $\nu $, then $i$ is aware of any coarsest refinement of $\nu $ belonging to $E$ and any coarsest refinement of $\nu $ belonging to $\neg E$, where  $\neg$ is the negation  in $\mathcal{RO}(\Omega,\sqsubseteq)$ from (\ref{NegDef}). 

\begin{remark} This definition of awareness of events---or equivalently binary questions---generalizes to a definition of awareness of arbitrary partitional questions, such as ``In what month was Ann born?'' Given a family $\{E_1,\dots,E_n\}$ of disjoint events such that $E_1\sqcup\dots\sqcup E_n=\Omega$, we  say that in $\omega$, agent $i$ is aware of $\{E_1,\dots,E_n\}$  if the following condition holds at $\omega$ and its refinements: for $1\leq k\leq n$, if $i$ is aware of a possibility $\nu $, then $i$ is aware of any coarsest refinement of $\nu $ belonging to $E_k$.\end{remark}

As for the requirement in part \ref{EventClosure} that $\mathcal{E}$ be closed under $\mathbf{A} $, note this requires that $\mathbf{A} (E)$ is a regular open set, which is indeed the case. The quantification over $\omega'\sqsubseteq\omega$ in the definition of $\mathbf{A} $ guarantees persistence  for $\mathbf{A} (E)$ (recall Lemma \ref{ROLem}), while awareness persistence and refinability guarantee refinability for $\mathbf{A} (E)$.\footnote{Conversely, in \textit{standard} frames as in Definition \ref{PossFrameAware}, the requirement that $\mathcal{RO}(\Omega,\sqsubseteq)$ be closed under $\mathbf{A}$ implies awareness persistence and awareness refinability.}

\begin{restatable}{lemma}{ROclosure}\label{ROclosure} Let $(\Omega,\sqsubseteq)$  be a poset and $\mathcal{A} :\Omega\to\wp(\Omega)$ satisfy awareness persistence and awareness refinability. Then for any $E\in\mathcal{RO}(\Omega,\sqsubseteq)$, we have $\mathbf{A} (E)\in\mathcal{RO}(\Omega,\sqsubseteq)$.
\end{restatable}
\noindent Then defining \textit{unawareness} by $\mathbf{U} (E)=\neg \mathbf{A} (E)$,  we have $\mathbf{U} (E)\in \mathcal{RO}(\Omega,\sqsubseteq)$ as well. Lemma \ref{ROclosure} implies that given awareness persistence and refinability, Definition \ref{PossFrameAware}.\ref{EventClosure} holds automatically when $\mathcal{E}=\mathcal{RO}(\Omega,\sqsubseteq)$.

Finally, the condition that $\mathscr{F}$ is standard is simply the condition that awareness of a possibility $\nu$  implies awareness of the distinction between $\mathord{\downarrow}\nu$ and $\neg\mathord{\downarrow}\nu$, which reduces to the condition that if $i$ is aware of $\nu$ and $\nu'$, then $i$ is also aware of any coarsest refinements of $\nu'$ that are incompatible with $\nu$: if $\nu,\nu'\in\mathcal{A} (\omega)$, then $\mathrm{max}(\{\nu^* \in\mathord{\downarrow}\nu'\mid \nu^*\,\bot\,\nu\})\subseteq \mathcal{A} (\omega)$. Standardness implies the equivalence given in (\ref{InduceAwareCorr}) above.

\begin{lemma} For any standard possibility frame with awareness $\mathscr{F}=(\Omega,\sqsubseteq, \mathcal{E} , \mathcal{A})$ and $\omega,\nu\in \Omega$, we have $\nu\in \mathcal{A} (\omega)$ if and only if $\omega\in\mathbf{A} (\mathord{\downarrow}\nu)$.
\end{lemma}
\begin{proof} The left-to-right direction is just the definition of standardness. For the right-to-left direction, if $\omega\in\mathbf{A} (\mathord{\downarrow}\nu)$, then since $m\in\mathcal{A} (\omega)$ by awareness nonvacuity, it follows from the definition of $\mathbf{A} $ that in $\omega$, $i$ is aware of the coarsest refinement of $m$ belonging to $\mathord{\downarrow}\nu$, which is $\nu$ itself, so $\nu\in \mathcal{A} (\omega)$.\end{proof}

We now give our first two examples of using possibility frames with awareness for modeling. For
simplicity, we continue to concentrate on the awareness of a single agent; but one can enrich each of our examples
to a multi-agent example by adding additional possibilities and correspondences for other agents' awareness.

\begin{example}\label{WatsonExample} We begin with perhaps the simplest example of an interesting possibility frame with awareness, formalizing a story discussed in \citealt{Geanakoplos1989}, \citealt{Modica1994}, and \citealt{Modica1999}. The story concerns Sherlock Holmes's assistant, Watson: if Watson hears a dog bark, then he will know---and hence be aware of---the event of the dog barking; but if he does not hear the dog bark, then he will not even be aware (at least at the relevant time) of the distinction between the dog barking vs.~not barking.  We formalize this using the frame in Figure \ref{WatsonFig}: $\Omega=\{\maximum , b,\overline{b}\}$; the partial order $\sqsubseteq$ is depicted by the arrows, so, e.g., we have  $b\sqsubseteq \maximum $ (arrows point toward more refined possibilities); the awareness correspondence for Watson, whom we will call $i$ for short, is given by $\mathcal{A} (\maximum )=\{\maximum \}$, $\mathcal{A} (b)=\Omega$, and $\mathcal{A} (\overline{b})=\{\maximum \}$; and $\mathcal{E}=\mathcal{RO}(\Omega,\sqsubseteq)$. It is easy to check that this is a standard possibility frame with awareness.

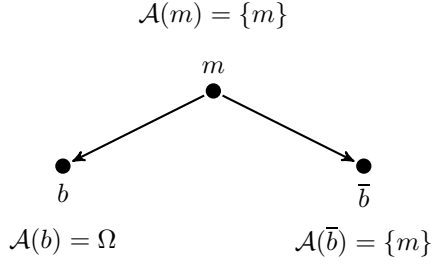
\begin{figure}[h]
\begin{center}
\begin{tikzpicture}[yscale=1, ->,>=stealth',shorten >=1pt,shorten <=1pt, auto,node
distance=2cm,thick,every loop/.style={<-,shorten <=1pt}]
\tikzstyle{every state}=[fill=gray!20,draw=none,text=black]
\node[circle,draw=black!100,fill=black!100, label=above:$\maximum $,inner sep=0pt,minimum size=.175cm] (0) at (0,0) {{}};
\node[circle,draw=black!100,fill=black!100, label=below:$b$,inner sep=0pt,minimum size=.175cm] (00) at (-2,-1) {{}};
\node[circle,draw=black!100,fill=black!100, label=below:$\overline{b}$,inner sep=0pt,minimum size=.175cm] (01) at (2,-1) {{}};

\node at (0,1) {{$\mathcal{A} (\maximum )=\{\maximum \}$}};
\node at (-2,-2) {{$\mathcal{A} (b)=\Omega$}};
\node  at (2,-2) {{$\mathcal{A} (\overline{b})=\{\maximum \}$}};

\path (0) edge[->] node {{}} (00);
\path (0) edge[->] node {{}} (01);

\end{tikzpicture}
\end{center}\label{WatsonFig}
\caption{A possibility frame with awareness representing Watson in the story of \citealt{Geanakoplos1989}. Refinement arrows implied by reflexivity are not drawn.}
\end{figure}

Let $Barks =\{b\}$; this set satisfies persistence and refinability, so by Lemma \ref{ROLem} it is an event in $\mathcal{RO}(\Omega,\sqsubseteq)$. Now it is immediate from the definition of $\mathbf{A} $ that when $i$ is aware of all possibilities, $i$ is also aware of all events. Hence in possibility $b$, $i$ is aware of $Barks$ and $\neg Barks$. By contrast, in $\overline{b}$, $i$ is unaware of these events, for the following reason: although in $\overline{b}$, $i$ is aware of the coarsest possibility $\maximum $, we have $\maximum \not\in Barks$ and $\maximum \not\in\neg Barks$ (the latter because $b\sqsubseteq\maximum $ and $b\in Barks$); then since in $\overline{b}$, $i$ is \textit{only} aware of $\maximum $, $i$ is not aware of the coarsest refinement of $\maximum $ belonging to $Barks$ (namely, $b$) or of the coarsest refinement of $\maximum $ belonging to $\neg Barks$ (namely, $\overline{b}$). In short, in $\overline{b}$, $i$ is not aware of the distinction $Barks$ vs.~$\neg Barks$. 

Moreover, in $\overline{b}$, $i$ is unaware of his unawareness of $Barks$. The reason is similar to the above: although in $\overline{b}$, $i$ is aware of $\maximum $, we have $\maximum \not\in \mathbf{A} (Barks) = \{b\}$ and $\maximum \not\in \neg \mathbf{A} (Barks) = \{\overline{b}\}$ (the latter because $b\sqsubseteq\maximum $ and $b\in \mathbf{A} (Barks)$); then since in $\overline{b}$, $i$ is \textit{only} aware of $\maximum $, $i$ is not aware of the coarsest refinement of $\maximum $ belonging to   $\mathbf{A} (Barks) $ (namely, $b$) or of the coarsest refinement of $\maximum $ belonging to $\neg\mathbf{A}  (Barks)$ (namely, $\overline{b}$). In short, in $\overline{b}$, $i$ is not aware of the distinction $\mathbf{A} (Barks)$ vs.~$\neg \mathbf{A} (Barks)$. 

Finally, we must continue, as stressed in Section \ref{Intro}, to avoid conflating events with sentences. For example, since we can write $\Omega=Barks\sqcup\neg Barks$, where $\sqcup$ is the join in the Boolean algebra $\mathcal{RO}(\Omega,\sqsubseteq)$, does it follow from Watson's being aware of $\Omega$ that he is aware of $Barks$? As noted in Section \ref{Intro}, it does not. $\Omega$ is the trivial event, not to be confused with the linguistic item `$Barks\sqcup \neg Barks$' of the analyst's language. Surely if Watson were aware of a sentence in a language that embeds `$Barks$', then he would be aware of `$Barks$'. But Watson's awareness of the trivial event $\Omega$ does not imply any such awareness of a sentence.
\end{example}

\begin{example}\label{GameEx} Consider a game in which a column player is aware that the row player can move up or down but is unaware that the row player has a third move, middle. Informally, such a situation is represented by the game matrix at the top of Figure \ref{GameFrame} in which the middle row is greyed out. Formally, we can represent the unawareness of the column player, whom we call $i$, using the frame in Figure \ref{GameFrame}; that this is a standard possibility frame with awareness can be checked by hand or more quickly with the notebook cited in Section~\ref{Organization}. There are two games the players could play: game $G$ in which the row player only has two moves, represented by the subtree with root $g$, and game $\underline{G}$ in which the row player has three moves, represented by the subtree with root $\underline{g}$. In each colored state, $i$ is aware only of the red possibilities; note $i$'s awareness of $\underline{g}$ is in effect just awareness of the possibility of \textit{not playing $G$}, without any awareness of further refinements of that possibility.  But before computing $i$'s awareness of events, one should become comfortable with the treatment of `not' and `or' coming from Theorem \ref{Tarski2}. For example, although the partial possibility $\underline{\ell}$ does not belong to the event $\underline{Middle} = \mathord{\downarrow}\{\underline{lm},\underline{rm}\}$ of the row player playing middle in $\underline{G}$, we have $\underline{\ell}\not\in \neg \underline{Middle}$, since $\underline{\ell}$ is refined by $\underline{\ell m}$ and $\underline{\ell m}\in \underline{Middle}$. Also note that where $\underline{Up}=\mathord{\downarrow}\{\underline{lu},\underline{ru}\}$ and $\underline{Down}=\mathord{\downarrow}\{\underline{ld},\underline{rd}\}$, we have $\underline{\ell}\in \underline{Up}\sqcup \underline{Middle}\sqcup \underline{Down}$, despite the fact that $\underline{\ell}$ does not belong to the union of these events; this is because every proper refinement of $\underline{\ell}$ does belong to the union. Thus, when dealing with partial possibilities, one must resist the temptation to interpret `not' and `or' using set-theoretic complement and union.

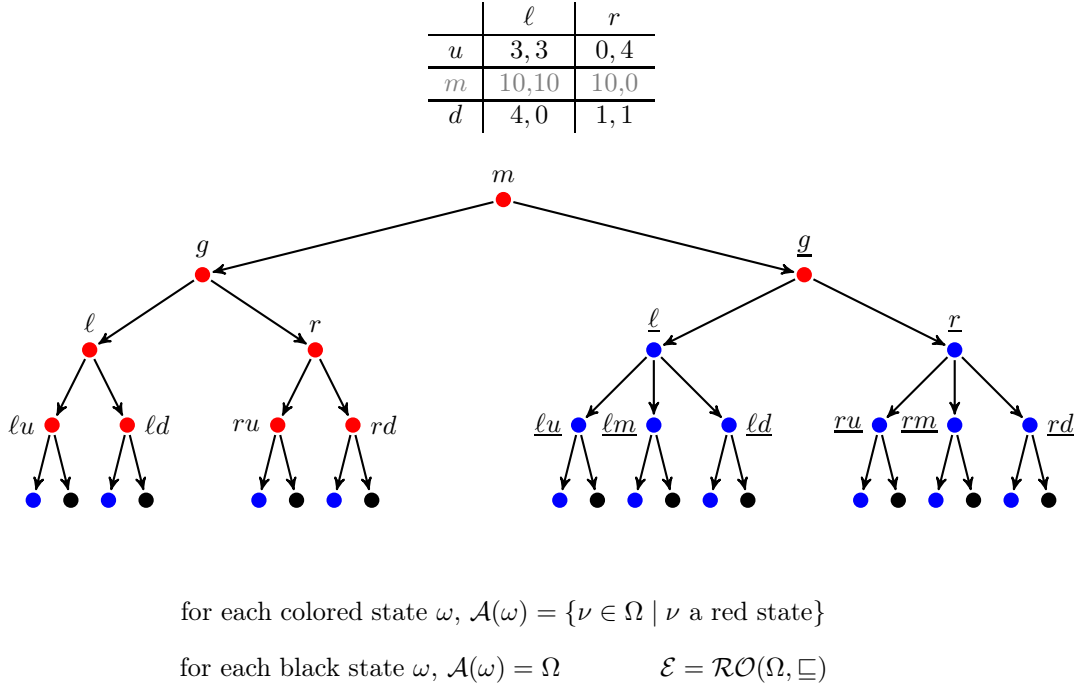
\begin{figure}[h]
\begin{center}
\begin{tabular}{c|c|c}
& $\ell$ & $r$ \\
\hline
$u$ & $3,3$ & $0,4$ \\
\hline
$\textcolor{gray}{m}$ & \textcolor{gray}{10,10} & \textcolor{gray}{10,0} \\
\hline
$d$ & $4,0$ & $1,1$ 
\end{tabular}
\end{center}
\begin{center}
\begin{tikzpicture}[yscale=1, ->,>=stealth',shorten >=1pt,shorten <=1pt, auto,node
distance=2cm,thick,every loop/.style={<-,shorten <=1pt}]
\tikzstyle{every state}=[fill=gray!20,draw=none,text=black]

\node[circle,draw=red!100,fill=red!100, label=above:$\maximum $,inner sep=0pt,minimum size=.175cm] (epsilon) at (4,1) {{}};

\node[circle,draw=red!100,fill=red!100, label=above:$g$,inner sep=0pt,minimum size=.175cm] (g) at (0,0) {{}};
\node[circle,draw=red!100,fill=red!100, label=above:$\ell$,inner sep=0pt,minimum size=.175cm] (l) at (-1.5,-1) {{}};
\node[circle,draw=red!100,fill=red!100, label=left:$\ell u$,inner sep=0pt,minimum size=.175cm] (lu) at (-2,-2) {{}};
\node[circle,draw=blue!100,fill=blue!100,inner sep=0pt,minimum size=.175cm] (lu1) at (-2.25,-3) {{}};
\node[circle,draw=black!100,fill=black!100,inner sep=0pt,minimum size=.175cm] (lu2) at (-1.75,-3) {{}};
\node[circle,draw=red!100,fill=red!100, label=right:$\ell d$,inner sep=0pt,minimum size=.175cm] (ld) at (-1,-2) {{}};
\node[circle,draw=blue!100,fill=blue!100,inner sep=0pt,minimum size=.175cm] (ld1) at (-1.25,-3) {{}};
\node[circle,draw=black!100,fill=black!100,inner sep=0pt,minimum size=.175cm] (ld2) at (-.75,-3) {{}};
\node[circle,draw=red!100,fill=red!100, label=above:$r$,inner sep=0pt,minimum size=.175cm] (r) at (1.5,-1) {{}};
\node[circle,draw=red!100,fill=red!100, label=left:$r u$,inner sep=0pt,minimum size=.175cm] (ru) at (1,-2) {{}};
\node[circle,draw=blue!100,fill=blue!100, label=left:,inner sep=0pt,minimum size=.175cm] (ru1) at (.75,-3) {{}};
\node[circle,draw=black!100,fill=black!100, label=left:,inner sep=0pt,minimum size=.175cm] (ru2) at (1.25,-3) {{}};
\node[circle,draw=red!100,fill=red!100, label=right:$rd$,inner sep=0pt,minimum size=.175cm] (rd) at (2,-2) {{}};
\node[circle,draw=blue!100,fill=blue!100, label=left:,inner sep=0pt,minimum size=.175cm] (rd1) at (1.75,-3) {{}};
\node[circle,draw=black!100,fill=black!100, label=left:,inner sep=0pt,minimum size=.175cm] (rd2) at (2.25,-3) {{}};

\path (g) edge[->] node {{}} (l);
\path (l) edge[->] node {{}} (lu);
\path (lu) edge[->] node {{}} (lu1);
\path (lu) edge[->] node {{}} (lu2);
\path (l) edge[->] node {{}} (ld);
\path (ld) edge[->] node {{}} (ld1);
\path (ld) edge[->] node {{}} (ld2);
\path (g) edge[->] node {{}} (r);
\path (r) edge[->] node {{}} (ru);
\path (r) edge[->] node {{}} (rd);
\path (ru) edge[->] node {{}} (ru1);
\path (ru) edge[->] node {{}} (ru2);
\path (rd) edge[->] node {{}} (rd1);
\path (rd) edge[->] node {{}} (rd2);

\node[circle,draw=red!100,fill=red!100, label=above:$\underline{g}$,inner sep=0pt,minimum size=.175cm] (g') at (8,0) {{}};
\node[circle,draw=blue!100,fill=blue!100, label=above:$\underline{\ell}$,inner sep=0pt,minimum size=.175cm] (l') at (6,-1) {{}};
\node[circle,draw=blue!100,fill=blue!100, label=left:$\underline{\ell u}$,inner sep=0pt,minimum size=.175cm] (lu') at (5,-2) {{}};
\node[circle,draw=blue!100,fill=blue!100,inner sep=0pt,minimum size=.175cm] (lu1') at (4.75,-3) {{}};
\node[circle,draw=black!100,fill=black!100,inner sep=0pt,minimum size=.175cm] (lu2') at (5.25,-3) {{}};
\node[circle,draw=blue!100,fill=blue!100, label=left:$\underline{\ell m}$,inner sep=0pt,minimum size=.175cm] (lm') at (6,-2) {{}};
\node[circle,draw=blue!100,fill=blue!100,inner sep=0pt,minimum size=.175cm] (lm1') at (5.75,-3) {{}};
\node[circle,draw=black!100,fill=black!100,inner sep=0pt,minimum size=.175cm] (lm2') at (6.25,-3) {{}};
\node[circle,draw=blue!100,fill=blue!100, label=right:$\underline{\ell d}$,inner sep=0pt,minimum size=.175cm] (ld') at (7,-2) {{}};
\node[circle,draw=blue!100,fill=blue!100,inner sep=0pt,minimum size=.175cm] (ld1') at (6.75,-3) {{}};
\node[circle,draw=black!100,fill=black!100,inner sep=0pt,minimum size=.175cm] (ld2') at (7.25,-3) {{}};
\node[circle,draw=blue!100,fill=blue!100, label=above:$\underline{r}$,inner sep=0pt,minimum size=.175cm] (r') at (10,-1) {{}};
\node[circle,draw=blue!100,fill=blue!100, label=left:$\underline{ru}$,inner sep=0pt,minimum size=.175cm] (ru') at (9,-2) {{}};
\node[circle,draw=blue!100,fill=blue!100, label=left:,inner sep=0pt,minimum size=.175cm] (ru1') at (8.75,-3) {{}};
\node[circle,draw=black!100,fill=black!100, label=left:,inner sep=0pt,minimum size=.175cm] (ru2') at (9.25,-3) {{}};
\node[circle,draw=blue!100,fill=blue!100, label=left:$\underline{rm}$,inner sep=0pt,minimum size=.175cm] (rm') at (10,-2) {{}};
\node[circle,draw=blue!100,fill=blue!100, label=left:,inner sep=0pt,minimum size=.175cm] (rm1') at (9.75,-3) {{}};
\node[circle,draw=black!100,fill=black!100, label=left:,inner sep=0pt,minimum size=.175cm] (rm2') at (10.25,-3) {{}};
\node[circle,draw=blue!100,fill=blue!100, label=right:$\underline{rd}$,inner sep=0pt,minimum size=.175cm] (rd') at (11,-2) {{}};
\node[circle,draw=black!100,fill=black!100, label=left:,inner sep=0pt,minimum size=.175cm] (rd1') at (11.25,-3) {{}};
\node[circle,draw=blue!100,fill=blue!100, label=left:,inner sep=0pt,minimum size=.175cm] (rd2') at (10.75,-3) {{}};

\path (g') edge[->] node {{}} (l');
\path (l') edge[->] node {{}} (lu');
\path (lu') edge[->] node {{}} (lu1');
\path (lu') edge[->] node {{}} (lu2');
\path (l') edge[->] node {{}} (lm');
\path (lm') edge[->] node {{}} (lm1');
\path (lm') edge[->] node {{}} (lm2');
\path (l') edge[->] node {{}} (ld');
\path (ld') edge[->] node {{}} (ld1');
\path (ld') edge[->] node {{}} (ld2');
\path (g') edge[->] node {{}} (r');
\path (r') edge[->] node {{}} (ru');
\path (ru') edge[->] node {{}} (ru1');
\path (ru') edge[->] node {{}} (ru2');
\path (r') edge[->] node {{}} (rm');
\path (rm') edge[->] node {{}} (rm1');
\path (rm') edge[->] node {{}} (rm2');
\path (r') edge[->] node {{}} (rd');
\path (rd') edge[->] node {{}} (rd1');
\path (rd') edge[->] node {{}} (rd2');

\path (epsilon) edge[->] node {{}} (g);
\path (epsilon) edge[->] node {{}} (g');

\node at (4,-4.5) {{for each colored state $\omega$, $\mathcal{A} (\omega)= \{\nu\in\Omega\mid \nu\mbox{ a red state}\}$}};

\node at (2.23,-5.25) {{for each black state $\omega$, $\mathcal{A} (\omega)=\Omega$}};

\node at (7.2,-5.25) {{$\mathcal{E}=\mathcal{RO}(\Omega,\sqsubseteq)$}};

\end{tikzpicture}
\end{center}
\caption{A possibility frame with awareness representing a column player's unawareness that the row player has an extra move $m$. Refinement arrows implied by reflexivity or transitivity are not drawn.}\label{GameFrame}
\end{figure}

Turning to awareness, observe that at each of the colored states $\omega$, it is not settled that $i$ is aware of $\underline{Middle}$:  $\omega\not\in\mathbf{A} ( \underline{Middle})$. For in the colored states,  $i$ is aware of $\underline{g}$, $\underline{g}\not\in \underline{Middle}$, and $\underline{g}\not\in \neg \underline{Middle}$ (since $\underline{g}$ is refined by $\underline{\ell m}$, which belongs to $\underline{Middle}$), yet $i$ is not aware of any coarsest refinement of $\underline{g}$ that belongs to $\underline{Middle}$, namely $\underline{\ell m}$ or $\underline{r m}$. Hence each colored \textit{leaf} $\omega$ of the tree settles that $i$ is \textit{not} aware of $\underline{Middle}$: $\omega\in \neg\mathbf{A} (\underline{Middle})$. By contrast, in the black leaves, $i$ is aware of every event, in virtue of being aware of every possibility. It follows that at each colored leaf $\omega$ of the tree, $i$ is unaware of her unawareness of $\underline{Middle}$: $\omega\in \neg\mathbf{A} \neg\mathbf{A} (\underline{Middle})$. For in the colored states, $i$ is aware of possibilities that are refined by colored leaves and black leaves, yet $i$ is unaware of all leaves.  Thus, at each colored leaf, $i$ is not only unaware of the distinction $\underline{Middle}$ vs.~$\neg \underline{Middle}$ but also of the distinction $\mathbf{A}  (\underline{Middle})$ vs.~$\neg\mathbf{A}  (\underline{Middle})$. Yet at all colored states, $i$ is aware of the events $Up=\mathord{\downarrow}\{\ell u,ru\}$ and $Down =\mathord{\downarrow}\{\ell d,rd\}$ of playing up and down in $G$, since for every red state $\omega$, every coarsest refinement of $\omega$ in $Up$, $\neg Up$, $Down$, and $\neg Down$ is itself a red state.

A question raised about the model in Figure \ref{GameFrame} is whether we should delete the two possibilities refining $\underline{\ell m}$ and then turn $\underline{\ell m}$ black (resulting in another frame satisfying the constraints of Definition \ref{PossFrameAware}), representing full awareness in state $\underline{\ell m}$ (and similarly for $\underline{rm}$): for in a state where the row player plays middle, doesn't the column player necessarily observe this? And isn't the column player therefore aware of $\underline{Middle}$ (as in the current black refinement of $\underline{\ell m}$)? In games with imperfect information, the column player might not observe the row player's move; but even if one assumes full observability, the analyst can assign zero probability to the event consisting of the singleton set of the blue possibility below $\underline{\ell m}$, rather than deleting that possibility from the model. One argument in favor of not deleting the possibility is that an agent can be aware of such a logical possibility (as the agent is in the black states) even if they are certain that it will not obtain (recall our intuitive explanation of $\omega'\in\mathcal{A}(\omega)$ at the beginning of Section \ref{ModelSection}).
\end{example}

We close this section with a simple example of an infinite possibility frame with awareness.

\begin{example}\label{InfiniteEx} Let $2^{\leq\omega}$ be the set of all finite or countably infinite binary strings (such as $0110$, etc.) ordered such that $\sigma\sqsubseteq \tau$ if $\tau$ is an initial segment of $\sigma$ (so $0110\sqsubseteq 011$, etc., reversed relative to the usual initial segment notation). This could represent all possibilities for finitely many or countably infinitely many flips of a coin. Now consider the possibility frame $(2^{\leq\omega},\sqsubseteq, \mathcal{RO}(2^{\leq\omega},\sqsubseteq))$, where each event in $  \mathcal{RO}(2^{\leq\omega})$ is equal to the union of an arbitrary set $\{\pi_i\}_{i\in I}$ of infinite strings and the set of all finite strings all of whose infinite extensions belong to $\{\pi_i\}_{i\in I}$. Equip the possibility frame with the awareness correspondence $\mathcal{A} $ defined by \[\mathcal{A} (\sigma)=\{\tau \in 2^{\leq\omega} \mid \tau\mbox{ a finite binary string}\}.\] Then $\mathcal{A} $ satisfies awareness nonvacuity (as the empty string is the maximal element of $(2^{\leq\omega},\sqsubseteq)$), expressibility (since $\mathcal{RO}(2^{\leq\omega},\sqsubseteq)$ contains all principal downsets $\mathord{\downarrow}\tau$), persistence and refinability (since $\mathcal{A} (\sigma)=\mathcal{A} (\tau)$ for all $\sigma,\tau \in 2^{\leq\omega}$), and joinability (since if $\mathrm{max}(E)\cap\mathord{\nu}$ and  $\mathrm{max}(E')\cap\mathord{\nu}$  contain only finite strings, then so does $\mathrm{max}((E\sqcup E')\cap\mathord{\downarrow}\nu)$).  This frame represents an agent who can conceive of any finite sequence of coin flips but cannot conceive of infinite sequences. Then the set of events of which the agent is aware forms an \textit{atomless} (recall Footnote \ref{AtomNote}) Boolean subalgebra (see Corollary \ref{SubalgCor}) of the atomic Boolean algebra $\mathcal{RO}(2^{\leq\omega},\sqsubseteq)$. The atoms of $\mathcal{RO}(2^{\leq\omega},\sqsubseteq)$ are the singleton sets of infinite binary strings, of which the agent is unaware.\end{example}

\subsection{Knowledge, belief, and awareness}\label{KnowledgeSection}

We now add knowledge and belief correspondences to our possibility frames with awareness. The basic distinction is that what the agent $i$ knows depends on the true information $i$ has received, whereas belief is subjective in the same sense as in subjective probability, which models belief quantitatively.\footnote{As sketched in Section~\ref{Conclusion}, we could use $p$-belief instead of belief, but for simplicity we use belief when introducing our model.} Take $\nu \in \mathcal{K} (\omega)$ (resp.~$\nu \in \mathcal{B} (\omega)$) to mean that every event that $i$ knows (resp.~believes) in $\omega$---or would know (resp.~believe) if made aware of the event---holds true in $\nu$. In this sense \textit{$\nu$ conforms to what $i$ knows} (resp.~believes).

\begin{definition}\label{EpistemicFrame} An \textit{epistemic possibility frame} is a tuple $\mathscr{F}=(\Omega,\sqsubseteq, \mathcal{E}, \mathcal{A},\mathcal{K},\mathcal{B})$ such that:
\begin{enumerate}
\item $(\Omega,\sqsubseteq, \mathcal{E}, \mathcal{A})$ is a possibility frame with awareness;
\item each $\mathcal{R} \in \{\mathcal{K} ,\mathcal{B} \} $ is a correspondence $\mathcal{R} :\Omega\to\wp(\Omega)$ satisfying the following for all $\omega,\omega',\nu\in\Omega$: 
\begin{enumerate}
\item $\mathcal{R} $-monotonicity: if $\omega'\sqsubseteq \omega$, then  $\mathcal{R} (\omega')\subseteq \mathcal{R} (\omega)$; 
\item $\mathcal{R} $-regularity: $\mathcal{R} (\omega)\in\mathcal{RO}(\Omega,\sqsubseteq)$;
\item $\mathcal{R} $-refinability: if $\nu \in \mathcal{R} (\omega)$, then $\exists \omega'\sqsubseteq \omega$ $\forall \omega''\sqsubseteq \omega'$ $\exists \nu'\sqsubseteq \nu$: $\nu'\in \mathcal{R} (\omega'')$;
\item epistemic factivity: $\omega\in \mathcal{K} (\omega)$;
\item doxastic consistency: $\mathcal{B} (\omega)\neq\varnothing$;
\item doxastic inclusion: $\mathcal{B} (\omega)\subseteq \mathcal{K} (\omega)$.
\end{enumerate}
\item\label{Eclosure} for each $\mathcal{R} \in \{\mathcal{K} ,\mathcal{B} \} $ and $E\in \mathcal{E}$, we have $\{\omega\in\Omega\mid \mathcal{R} (\omega)\subseteq E\}\in\mathcal{E}$.
\end{enumerate}
\noindent We call $\mathscr{F}$ \textit{standard} if the underlying possibility frame with awareness is standard.
\end{definition}

The interpretations of the first three conditions are as follows for knowledge; the belief interpretations are analogous. $\mathcal{R} $-monotonicity---in its equivalent formulation: if $\omega'\sqsubseteq \omega$, then $\nu \not\in \mathcal{R} (\omega)$ implies $\nu \not\in \mathcal{R} (\omega')$---says that if $\nu $ does not conform to $i$'s knowledge in $\omega$, then $\nu$ does not conform to $i$'s knowledge in any refinement $\omega'$ of $\omega$, since $i$ retains in $\omega'$ whatever knowledge she had in $\omega$. $\mathcal{R} $-regularity says that we can view the set of possibilities that conform to $i$'s knowledge as a genuine event that $i$ implicitly knows  (in fact, the strongest such event, given the definition of implicit knowledge below). Finally, $\mathcal{R} $-refinability says that if $\nu $ conforms to $i$'s knowledge in $\omega$, then there is a refinement $\omega'$ of $\omega$ that settles that \textit{some refinement of $\nu $ conforms to $i$'s knowledge}, in the sense that at every refinement $\omega''$ of $\omega'$, some refinement of $\nu $ conforms to $i$'s knowledge in $\omega''$.\footnote{This condition from \citealt{Holliday2015} is weaker than the refinability condition in \citealt{Humberstone1981}. This weakening is useful for the representation theory of modal algebras by possibility frames (see Remark 2.39 of \citealt{Holliday2015}).} Together these conditions imply the following closure property of $\mathcal{RO}(\Omega,\sqsubseteq)$, which shows that it is possible to satisfy the closure property in part \ref{Eclosure} of Definition \ref{EpistemicFrame}.

\begin{restatable}{lemma}{ROclosureTwo}\label{ROclosure2} Let $(\Omega,\sqsubseteq)$  be a poset and $\mathcal{R} :\Omega\to\wp(\Omega)$ satisfy $\mathcal{R} $-monotonicity, $\mathcal{R} $-regularity, and $\mathcal{R} $-refinability. Then for any $E\in\mathcal{RO}(\Omega,\sqsubseteq)$, we have $\{\omega\in\Omega\mid \mathcal{R} (\omega)\subseteq E\}\in \mathcal{RO}(\Omega,\sqsubseteq)$. \end{restatable}
\noindent Thus, if $\mathcal{R}$ satisfies the listed properties, then Definition \ref{EpistemicFrame}.\ref{Eclosure} holds automatically when $\mathcal{E}=\mathcal{RO}(\Omega,\sqsubseteq)$.\footnote{A natural question is whether the converse of Lemma \ref{ROclosure2} holds, i.e., if for any $E\in\mathcal{RO}(\Omega,\sqsubseteq)$, we have ${\{\omega\in\Omega\mid \mathcal{R} (\omega)\subseteq E\}}\in \mathcal{RO}(\Omega,\sqsubseteq)$, does this imply the listed properties of $\mathcal{R}$? In fact, it implies slightly weaker properties of $\mathcal{R}$, as shown in Proposition 2.30 of \citealt{Holliday2015}, but the stronger properties can be assumed without loss of generality (Proposition 2.37 of \citealt{Holliday2015}).}

As in \citealt{Fagin1988}, one may take $\mathbf{L} (E)=\{\omega\in\Omega\mid \mathcal{K} (\omega)\subseteq E\}$ (resp.~$\{\omega\in\Omega\mid \mathcal{B} (\omega)\subseteq E\}$) to be the event of $i$ \textit{implicitly knowing} (resp. \textit{implicitly believing}) $E$ in the sense that $i$ would know (resp.~believe) $E$ if $i$ were aware of $E$. However, we will concentrate here on explicit knowledge $\mathbf{K} $ and belief $\mathbf{B} $:
\begin{eqnarray*}
\mathbf{K} (E)&=& \{\omega\in\Omega\mid \mathcal{K} (\omega)\subseteq E\mbox{ and } \omega\in\mathbf{A} (E)\};\\
\mathbf{B} (E)&=&\{\omega\in\Omega\mid \mathcal{B} (\omega)\subseteq E\mbox{ and }\omega\in \mathbf{A} (E)\}.\end{eqnarray*}
By Lemma \ref{ROclosure2} and the closure of $\mathcal{E}$ under binary intersection,  $\mathcal{E}$ is also closed under $\mathbf{K} $ and $\mathbf{B} $. Moreover, in the multi-agent generalization of our setup with correspondences $\mathcal{K}_i$ and $\mathcal{B}_i$ for each agent $i$, if $\mathcal{E}$ is closed under countably infinite intersections (e.g., if $\mathcal{E}=\mathcal{RO}(\Omega,\sqsubseteq)$), then it is closed under the usual operations of \textit{common knowledge} (\citealt[\S~2]{Aumann1999}) and \textit{common belief} (see, e.g., \citealt{Lismont1994}).

Conditions (d)--(f) of Definition \ref{EpistemicFrame} are the bare minimum constraints for knowledge and belief, familiar from the earliest formal models of knowledge and belief (\citealt{Hintikka2005}). Epistemic factivity implies that if $i$ knows $E$, then $E$ is true; doxastic consistency implies that an agent cannot believe $\varnothing$; and doxastic inclusion (which implies that if $\mathcal{K} (\omega)\subseteq E$, then $\mathcal{B} (\omega)\subseteq E$) implies that if $i$ knows $E$, then $i$ believes $E$. 

\begin{remark}\label{ExtraConditions} Particular frames used in applications may of course satisfy additional constraints, which may imply epistemic and doxastic introspection principles (see \citealt{Ding2019} for a study of when such principles matter for multi-agent reasoning). As usual, introspection principles for $\mathbf{A} $, $\mathbf{K} $, and $\mathbf{B} $ that quantify over events---e.g., for all events $E$, $\mathbf{K} (E)\subseteq\mathbf{K} (\mathbf{K} (E))$---immediately correspond to conditions on $\mathcal{A} $, $\mathcal{K} $, and $\mathcal{B} $ that quantify over events, just by unpacking the definitions of $\mathbf{A} $, $\mathbf{K} $, and $\mathbf{B} $. There is then a technical question, studied in a branch of modal logic known as correspondence theory (\citealt{Benthem1980}), about whether the conditions on $\mathcal{A} $, $\mathcal{K} $, and $\mathcal{B} $ that quantify over events---called second-order conditions---are equivalent to conditions on $\mathcal{A} $, $\mathcal{K} $, and $\mathcal{B} $ that only quantify over possibilities in $\Omega$---called first-order conditions. Correspondence theory for implicit knowledge and belief in possibility semantics is well understood (\citealt{Holliday2015}, \citealt{Yamamoto2017}, \citealt{Zhao2016}). To take one example, in epistemic possibility frames,  Positive Introspection for implicit knowledge---for all $E\in\mathcal{RO}(\Omega,\sqsubseteq)$, $\mathbf{L} (E)\subseteq \mathbf{L} (\mathbf{L} (E))$---holds if and only if  $\mathcal{K} $ is transitive (if $\omega'\in\mathcal{K} (\omega)$ and $\omega''\in\mathcal{K} (\omega')$, then $\omega''\in\mathcal{K} (\omega)$), just as in Kripke frames. However, we will not delve into correspondence theory for awareness or explicit knowledge and belief here, as our focus is instead on representation theory in Section \ref{Representation}. The example frames in this paper satisfy first-order conditions that are sufficient but not necessary for introspection principles. For example, it is sufficient for $\mathbf{A} (E)\subseteq \mathbf{A} (\mathbf{A} (E))$ to hold for all $E\in\mathcal{E}$ that for all $\omega\in\Omega$, if $\mathcal{A} (\omega)\neq \Omega$, then for all $\nu\in\mathcal{A} (\omega)$ and $\nu'\sqsubseteq \nu$, if $\mathcal{A} (\nu')\neq \Omega$, then for some $\omega'\sqsubseteq\omega$, we have $\mathcal{A} (\nu')=\mathcal{A} (\omega')$; and if this also holds with `if $\mathcal{A} (\omega)\neq \Omega$, then for all $\nu\in\mathcal{A} (\omega)$' replaced by `for all $\nu\in\mathcal{K} (\omega)$', then $\mathbf{A} (E)\subseteq \mathbf{K} (\mathbf{A} (E))$ for all $E\in\mathcal{E}$.\end{remark}

We now present two examples of using epistemic possibility frames for modeling.

\begin{example}\label{WatsonKnowledge} Let us add knowledge and belief to Example \ref{WatsonExample}. There are two pairs of knowledge and belief correspondences we might consider for Watson:
 \[\mbox{$\mathcal{K} (b)=\mathcal{B} (b)=\{b\}$ and $\mathcal{K} (\overline{b})=\mathcal{B} (\overline{b})=\mathcal{K} (\maximum )=\mathcal{B} (\maximum )=\Omega$};\] 
  \[\mbox{$\mathcal{K} '(b)=\mathcal{B} '(b)=\{b\}$ and $\mathcal{K} '(\overline{b})=\mathcal{B} '(\overline{b})=\{\overline{b}\}\mbox{ and }\mathcal{K} '(\maximum )=\mathcal{B} '(\maximum )=\Omega$.}\] 
 One can easily check that in both cases, all the conditions of Definition \ref{EpistemicFrame} are satisfied. Moreover, for any event $E$, we have $\mathbf{K} (E)=\mathbf{K} '(E)$ and $\mathbf{B} (E)=\mathbf{B} '(E)$. However, the primed pair of correspondences can be used to capture the idea that \textit{if only Watson were aware in $\overline{b}$ of the distinction between $Bark$ and $\neg Bark$}, then he would know and believe $\neg Bark$ in $\overline{b}$. In either case, the frame illustrates our reason for rejecting part of what Dekel et al.~\citeyearpar{Dekel1998}  call the axiom of Plausibility, which is one direction of the Modica-Rustichini definition of unawareness: $\mathbf{U}(E)\subseteq \neg \mathbf{K}(E)$ and $\mathbf{U}(E)\subseteq \neg \mathbf{K}\neg \mathbf{K}(E)$.  The first inclusion holds in our model by the definition of $\mathbf{K}$, but the second inclusion can fail in light of the distinction between awareness of events and awareness of sentences. First observe that the event $ \mathbf{U} (Barks)$ is \textit{unknowable}: it cannot be known at $b$ or $\maximum $, because it is not true at these states (recall Example \ref{WatsonExample}); and it cannot be known at $\overline{b}$, because although it is true at $\overline{b}$, Watson is not aware of $\mathbf{U}(Barks)$ in $\overline{b}$ (again recall Example \ref{WatsonExample}). Thus, $\neg \mathbf{K}  ( \mathbf{U} (Barks))=\Omega$. But of course $\mathbf{K} (\Omega)=\Omega$ (recall Section \ref{Intro}). So we have a violation of Plausibility at $\overline{b}$, as $\overline{b}\in \mathbf{U} (Barks)$ and $\overline{b}\not\in \neg \mathbf{K}\neg\mathbf{K} (\mathbf{U} (Barks))$. Here Watson is aware of and indeed knows the trivial event $\Omega$, and as \textit{analysts} we see that $\Omega=\neg\mathbf{K} (\mathbf{U} (Barks))$. But it is a fallacy, based on conflating awareness of events with awareness of sentences, to think that if $i$ is aware of $\Omega$, and $\Omega$ can be obtained from $E$ by applying some operations on sets, then $i$ must be aware of $E$. It is similar to the fallacy in thinking that if a student is aware of a number $n$, and $n$ can be obtained from a number $m$ by some mathematical operations, then the student must also be aware of $m$. For more on Plausibility, see Appendix~\ref{PlausibilityAppendix} and Fact~\ref{PossibilityFact}.\end{example}
 
 \begin{remark} We can now prove that in our model, unawareness is not definable in any way from knowledge and belief; so not only does the Modica-Rustichini definition of unawareness in terms of knowledge fail, but other definitions in terms of knowledge and/or belief also fail. It suffices to give two possibility frames that yield the same algebra of events and the same $\mathbf{K}$ and $\mathbf{B}$ operators but different $\mathbf{A}$ operators. Let $\mathscr{F}_1$ have the same three possibilities as in Example \ref{WatsonExample} but with $\mathcal{A}(\omega)=\{m\}$ and $\mathcal{K}(\omega)=\mathcal{B}(\omega)=\Omega$ for all $\omega\in\Omega$. Let $\mathscr{F}_2$ be just like $\mathscr{F}_1$ except that we change $\mathcal{A}$ so that $\mathcal{A}(\omega)=\Omega$ for all $\omega\in\Omega$. In both frames, the agent does not know or believe anything beyond the trivial event $\Omega$; but in $\mathscr{F}_1$, the agent has unawareness (e.g., $\omega\not\in \mathbf{A} (Barks)$ for all $\omega$), while in $\mathscr{F}_2$, the agent has full awareness ($\omega\in\mathbf{A}(E)$ for all possibilities $\omega$ and events~$E$). Thus, the awareness operator is not definable in terms of the knowledge and belief operators.\end{remark}
 
Finally, let us return to the example of the \textit{overconfident} agent with which we began in Section \ref{Intro}.

 \begin{example}\label{OverconfidentEx} A potential investor $i$ in a firm believes that he knows the firm is profitable, while being unaware of a sophisticated type of fraud that the firm is in fact using to cover up unprofitability. This is the case in  state $f_1$ in the possibility frame in Figure \ref{OverconfidentFig} with $\mathcal{E}=\mathcal{RO}(\Omega,\sqsubseteq)$ and the awareness, belief, and knowledge correspondences specified below. Intuitively, in the names for states, $p$ stands for \textit{profitability}, $b$ stands for \textit{belief} (in profitability), $u$ stands for \textit{unawareness} (of possibilities of fraud), and $f$ stands for \textit{fraud}. Then the correspondences are defined as follows:
 
 \begin{itemize}
\item for each colored state $\omega$, $\mathcal{A} (\omega)= \{\nu\in\Omega\mid \nu\mbox{ a red state}\}$;
\item for each black or gray state $\omega$, $\mathcal{A} (\omega)=\Omega$;
\item for each square or blue state $\omega$, $\mathcal{B} (\omega)=\mathord{\downarrow}pb$ and $\mathcal{K} (\omega)=\mathord{\downarrow}\{\omega, pb\}$;
\item for each black state $\omega$, $\mathcal{B} (\omega)= \{pb\overline{u}\}$ and $\mathcal{K} (\omega)=\{\omega,pb\overline{u}\}$;
\item for each diamond or green state $\omega$,  $\mathcal{B} (\omega)=\mathcal{K} (\omega)= \mathord{\downarrow}\{\nu\in\Omega\mid \nu\mbox{ a diamond state}\}$;
\item for each gray state $\omega$, $\mathcal{B} (\omega)=\mathcal{K} (\omega)= \{\nu\in\Omega\mid \nu\mbox{ a gray state}\}$;
\item $\mathcal{B} (p)=\rho(\underset{\omega \sqsubset p}{\bigcup}\mathcal{B} (\omega)) = \mathord{\downarrow}\{p,\overline{p}\overline{b}\}$ and $\mathcal{K} (p)=\rho(\underset{\omega \sqsubset p}{\bigcup}\mathcal{K} (\omega)) = \mathord{\downarrow}\{p,\overline{p}\overline{b}\}$;
\item $\mathcal{B} (\overline{p})=\rho(\underset{\omega \sqsubset \overline{p}}{\bigcup}\mathcal{B} (\omega)) = \mathord{\downarrow}\{p,\overline{p}\overline{b}\}$ and $\mathcal{K} (\overline{p})=\rho(\underset{\omega \sqsubset \overline{p}}{\bigcup}\mathcal{K} (\omega)) = \Omega$;
\item $\mathcal{B} (m)=\mathord{\downarrow}\{p,\overline{p}\overline{b}\}$ and $\mathcal{K} (m)=\Omega$.
\end{itemize}

\begin{figure}[h]

\begin{center}
\begin{tikzpicture}[yscale=1, ->,>=stealth',shorten >=1pt,shorten <=1pt, auto,node
distance=2cm,thick,every loop/.style={<-,shorten <=1pt}]
\tikzstyle{every state}=[fill=gray!20,draw=none,text=black]

\node[circle,draw=red!100,fill=red!100, label=above:$\maximum $,inner sep=0pt,minimum size=.175cm] (epsilon) at (4,1) {{}};

\node[circle,draw=red!100,fill=red!100, label=above:$p$,inner sep=0pt,minimum size=.175cm] (g) at (0,0) {{}};
\node[rectangle,draw=red!100,fill=red!100, label=left:$pb$,inner sep=0pt,minimum size=.175cm] (l) at (-1.5,-1) {{}};
\node[circle,draw=blue!100,fill=blue!100, label=below:$pbu$,inner sep=0pt,minimum size=.175cm] (lu) at (-2,-2) {{}};
\node[circle,draw=black!100,fill=black!100, label=below:$pb\overline{u}$,inner sep=0pt,minimum size=.175cm] (ld) at (-1,-2) {{}};
\node[diamond,draw=red!100,fill=red!100, label=right:$p\overline{b}$,inner sep=0pt,minimum size=.175cm] (r) at (1.5,-1) {{}};
\node[circle,draw=medgreen!100,fill=medgreen!100, label=below:$p\overline{b}u$,inner sep=0pt,minimum size=.175cm] (ru) at (1,-2) {{}};
\node[circle,draw=gray!100,fill=gray!100, label=below:$p\overline{b}\overline{u}$,inner sep=0pt,minimum size=.175cm] (rd) at (2,-2) {{}};
\path (g) edge[->] node {{}} (l);
\path (l) edge[->] node {{}} (lu);
\path (l) edge[->] node {{}} (ld);
\path (g) edge[->] node {{}} (r);
\path (r) edge[->] node {{}} (ru);
\path (r) edge[->] node {{}} (rd);
\node[circle,draw=red!100,fill=red!100, label=above:$\overline{p}$,inner sep=0pt,minimum size=.175cm] (g') at (8,0) {{}};
\node[rectangle,draw=red!100,fill=red!100, label=left:$\overline{p}b$,inner sep=0pt,minimum size=.175cm] (l') at (6,-1) {{}};

\node[circle,draw=blue!100,fill=blue!100, label=left:$\overline{p}bu$,inner sep=0pt,minimum size=.175cm] (lu') at (5.25,-2) {{}};
\node[circle,draw=blue!100,fill=blue!100,label=below:$f_1$,inner sep=0pt,minimum size=.175cm] (lu1') at (5,-3) {{}};
\node[circle,draw=blue!100,fill=blue!100,label=below:$\overline{f}_1$,inner sep=0pt,minimum size=.175cm] (lu2') at (5.5,-3) {{}};

\node[circle,draw=black!100,fill=black!100, label=right:$\overline{p}b\overline{u}$,inner sep=0pt,minimum size=.175cm] (ld') at (6.75,-2) {{}};
\node[circle,draw=black!100,fill=black!100,label=below:$f_2$,inner sep=0pt,minimum size=.175cm] (ld1') at (6.5,-3) {{}};
\node[circle,draw=black!100,fill=black!100,label=below:$\overline{f}_2$,inner sep=0pt,minimum size=.175cm] (ld2') at (7,-3) {{}};
\node[diamond,draw=red!100,fill=red!100, label=right:$\overline{p}\overline{b}$,inner sep=0pt,minimum size=.175cm] (r') at (10,-1) {{}};
\node[circle,draw=medgreen!100,fill=medgreen!100, label=left:$\overline{p}\,\overline{b}u$,inner sep=0pt,minimum size=.175cm] (ru') at (9.25,-2) {{}};
\node[circle,draw=medgreen!100,fill=medgreen!100, label=below:$f_3$,inner sep=0pt,minimum size=.175cm] (ru1') at (9,-3) {{}};
\node[circle,draw=medgreen!100,fill=medgreen!100, label=below:$\overline{f}_3$,inner sep=0pt,minimum size=.175cm] (ru2') at (9.5,-3) {{}};

\node[circle,draw=gray!100,fill=gray!100, label=right:$\overline{p}\overline{b}\overline{u}$,inner sep=0pt,minimum size=.175cm] (rd') at (10.75,-2) {{}};
\node[circle,draw=gray!100,fill=gray!100, label=below:$\overline{f}_4$,inner sep=0pt,minimum size=.175cm] (rd1') at (11,-3) {{}};
\node[circle,draw=gray!100,fill=gray!100, label=below:$f_4$,inner sep=0pt,minimum size=.175cm] (rd2') at (10.5,-3) {{}};

\path (g') edge[->] node {{}} (l');
\path (l') edge[->] node {{}} (lu');
\path (lu') edge[->] node {{}} (lu1');
\path (lu') edge[->] node {{}} (lu2');
\path (l') edge[->] node {{}} (ld');
\path (ld') edge[->] node {{}} (ld1');
\path (ld') edge[->] node {{}} (ld2');
\path (g') edge[->] node {{}} (r');
\path (r') edge[->] node {{}} (ru');
\path (ru') edge[->] node {{}} (ru1');
\path (ru') edge[->] node {{}} (ru2');
\path (r') edge[->] node {{}} (rd');
\path (rd') edge[->] node {{}} (rd1');
\path (rd') edge[->] node {{}} (rd2');

\path (epsilon) edge[->] node {{}} (g);
\path (epsilon) edge[->] node {{}} (g');

\end{tikzpicture}
\end{center}
\caption{The refinement structure of a possibility frame for Example \ref{OverconfidentEx}. Refinement arrows implied by reflexivity or transitivity are not drawn.}
\end{figure}
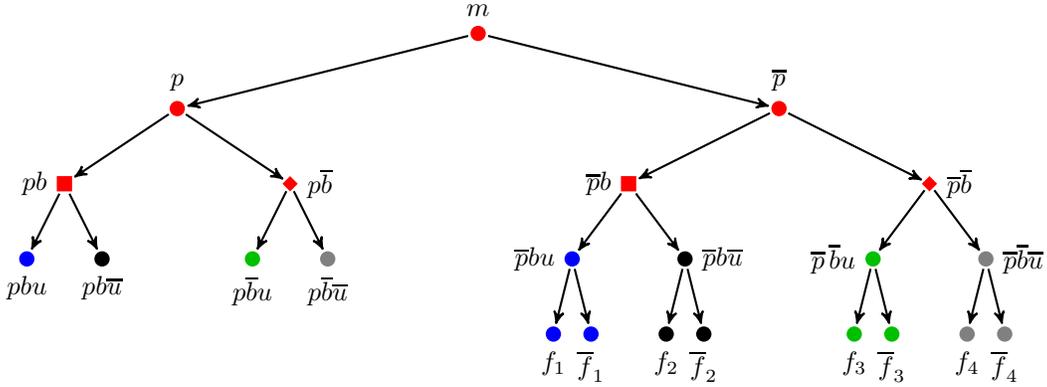\label{OverconfidentFig}

Let $Profit = \mathord{\downarrow}p$ and $Fraud =\{f_1,f_2,f_3,f_4\}$. Then we have the following:
\begin{itemize}
\item In the blue states, $i$ is unaware of the distinction between $Fraud$ and $\neg Fraud$ and believes $Profit$ and believes that he knows $Profit$ (note that $\mathcal{B} (\omega)=\mathord{\downarrow}pb$ for each blue state $\omega$, and in $pb$, $i$ knows $Profit$).
\item In the green states, $i$ is unaware of the distinction between $Fraud$ and $\neg Fraud$ and is undecided (in belief) and uncertain (in knowledge) about $Profit$.
\item In the black states, $i$ is aware of the distinction between $Fraud$ and $\neg Fraud$ but still believes $Profit$ and believes he knows $Profit$ (note that  $\mathcal{B} (\omega)= \{pb\overline{u}\}$ for each black state $\omega$, and in $pb\overline{u}$, $i$ knows~$Profit$).
\item In the gray states, $i$ is aware of the distinction between $Fraud$ and $\neg Fraud$ and is undecided and uncertain about $Profit$.
\end{itemize}
In $f_1$, $\overline{f}_1$, $f_2$, and $\overline{f}_2$, agent $i$ is overconfident in $Profit$: $i$ believes he knows $Profit$, but he does not, since $Profit$ does not obtain in these states.  As analysts---or as other agents interacting with $i$---we might assign low probability to $f_2$ and $\overline{f}_2$, reflecting the view that $i$ is unlikely to be overconfident when $i$ is aware of the possibility of the sophisticated type of fraud. In general, say that a state $\omega$ determines that agent $i$ has \textit{information-based beliefs}\footnote{Similarly, one could define \textit{information-based $p$-beliefs} by replacing $\nu\in \mathcal{B} (\omega')$ in the displayed condition with $i$'s assigning probability greater than $1-p$ to the event $\mathord{\downarrow}\nu$, assuming an extension of our model with probability as sketched in Section~\ref{Conclusion}.} if for all $\omega'\sqsubseteq\omega$ and $\nu\in \Omega$, if in $\omega'$, $i$ is aware of $\nu$ and $i$'s information is in fact consistent with $\nu$, then in $\omega'$, $i$ does not mistakenly believe he can rule out $\nu$: \[\mbox{if $\nu\in \mathcal{A} (\omega')$ and $\nu\in\mathcal{K} (\omega')$, then $\nu\in\mathcal{B} (\omega')$.}\] If $\omega$ determines that $i$ has information-based beliefs, then $\omega$ determines that $i$ can only have false beliefs if he is unaware of some possibilities that are in fact consistent with his information. In the above example, $i$ has information-based beliefs at all leaves of the tree except for $f_2$ and $\overline{f}_2$. In principle, to what extent false beliefs in a population of agents are correlated with unawareness of possibilities consistent with their information, as distinguished from mistakes in reasoning, exposure to misleading evidence, etc., could be investigated experimentally, though we will not attempt to describe an experimental protocol here.

To sum up, in contrast to models that define awareness in terms of knowledge (recall Section~\ref{Intro}), as in \citealt{Modica1994,Modica1999}, \citealt{Heifetz2006}, and \citealt{Li2009}, we are able to model an agent who is aware of $Profit$ while also being overconfident in $Profit$; and we are able to relate the agent's overconfidence in $Profit$ to his unawareness of $Fraud$.\end{example}

 \begin{remark} Examples \ref{WatsonKnowledge} and \ref{OverconfidentEx} satisfy all the principles of Stalnaker's \citeyearpar[p.~179]{Stalnaker2006} joint logic of knowledge and belief, with negative introspection for belief suitably modified to allow for unawareness: $\textbf{K} (E)\subseteq \textbf{K} (\textbf{K} (E))$ and $\textbf{B} (E)\subseteq \textbf{K} (\textbf{B} (E))$  (Positive Introspection); $\neg\mathbf{B} (E)\cap \mathbf{A} \neg\mathbf{B} (E)\subseteq \mathbf{K} \neg\mathbf{B} (E)$ (Weak Negative Introspection for Belief); and $\mathbf{B} (E)\subseteq\mathbf{B} (\mathbf{K} (E))$ (Strong Belief).\end{remark}

Finally, let us relate our model to the impossibility theorem concerning unawareness from \citealt{Dekel1998}, reviewed in detail in Appendix \ref{DLR}. We already discussed Dekel et al.'s axiom of Plausibility in Example~\ref{WatsonKnowledge}. One can check\footnote{This can be checked by hand for Example \ref{WatsonKnowledge} and using the notebook cited in Section \ref{Organization} for Example \ref{OverconfidentEx}.} that the structures $(\mathcal{E},\subseteq, \mathbf{U} ,\mathbf{K} ,\neg)$ arising from Examples \ref{WatsonKnowledge} and \ref{OverconfidentEx} satisfy all of Dekel et al.~other axioms (in particular, what they call AU Introspection and KU Introspection) and the following weakening of Plausibility:
 \begin{itemize}
 \item[] Nontrivial Plausibility: $\mathbf{U}(E)\subseteq \neg \mathbf{K}(E)$, and if $\neg \mathbf{K}(E)\neq \Omega$, then $\mathbf{U}(E)\subseteq \neg \mathbf{K}\neg \mathbf{K}(E)$.
 \end{itemize}
 Thus, a  slight weakening of Plausibility leads us from  impossibility to the following possibility result. 
 
 \begin{fact}\label{PossibilityFact} There are epistemic possibilities frames such that $(\mathcal{E},\subseteq, \mathbf{U} ,\mathbf{K} ,\neg)$ satisfies the axioms of \citealt[Theorem 1(i)]{Dekel1998} when Plausibility is replaced by Nontrivial Plausibility.\footnote{\label{OtherPrinciples}The frames from Examples \ref{WatsonKnowledge} and \ref{OverconfidentEx} also satisfy the following, building on nomenclature of  \citealt[Prop.~3]{Heifetz2006}: if $\mathbf{K} (E)\neq\varnothing$, then $\mathbf{A} (\mathbf{K} (E))=\mathbf{A} (E)$ (Nontrivial AK-Self Reflection); $\mathbf{A} (\mathbf{A} (E))= \mathbf{A} (E)$ (AA-Self Reflection); and $\mathbf{K} (\mathbf{A} (E)) = \mathbf{A} (E)$ (A-Introspection) (recall Remark \ref{ExtraConditions}). However, we think it would be reasonable for at least the left-to-right inclusions in each of these principles to be violated in other examples, given the distinction between events and sentences we have stressed (e.g., from the facts that $\omega\in\mathbf{A} (F)$ and $F=\mathbf{A} (E)$, it should not be required that $\omega\in\mathbf{A} (E)$, since the agent need not conceive of the event $F$---which has no syntactic structure---in terms of anyone's awareness of $E$). On similar grounds, the second part of Nontrivial Plausibility can also be reasonably doubted as a universal requirement in all cases.}
 \end{fact}
 
\section{Representation}\label{Representation}

Having seen how to model some concrete examples involving awareness, knowledge, and belief using possibility frames, we now identify exactly the class of examples that can be so represented. Such an example is specified abstractly by a Boolean algebra of events equipped with awareness, knowledge, and belief operators.

\begin{definition}\label{AwareAlg} An \textit{epistemic awareness algebra} is a tuple $\mathbb{A}=(\mathbb{B}, A, K, B)$ where $\mathbb{B}$ is a Boolean algebra and $A $, $K $, and $B $ are unary operations on $\mathbb{B}$ such that for all $a,b\in \mathbb{B}$ and $\Box \in \{K ,B \}$:
\begin{itemize}
\item\label{AwareOfTop} $A  1=1$ (tautology), $A  a=A \neg a$ (symmetry), and $A a\sqcap A b\leq A (a\sqcap b)$ (agglomeration);
\item $K 1=1$ (knowledge necessitation) and $\Box  a\sqcap \Box  b\leq \Box (a\sqcap b)$ (knowledge/belief agglomeration);
\item if $a\leq b$, then $\Box  a\sqcap A  b\leq \Box  b$ (awareness-restricted monotonicity);
\item $K a\leq a$ (factivity of knowledge) and $B 0=0$ (consistency of belief);
\item $K a\leq B a$ (knowledge-belief entailment) and $B a\leq A a$ (belief-awareness entailment).
\end{itemize}
As usual, an \textit{isomorphism} between epistemic awareness algebras $\mathbb{A}$ and $\mathbb{A}'$ is a bijection from the carrier set of $\mathbb{A}$ to that of $\mathbb{A}'$ that respects the lattice orders and operations: $a\leq b$ if and only if $f(a)\leq' f(b)$; and for each $\triangle \in \{A ,K ,B \} $, we have $f(\triangle (a))=\triangle' (f(a))$.
\end{definition}
Some immediate consequences are that $B 1=1$ (belief necessitation) and $K a\leq A a$ (knowledge-awareness entailment). It also follows that awareness is closed under not only meets but also joins.

\begin{lemma}\label{JoinClosure} For any $a,b\in B$, $ A a \sqcap  A  b\leq  A  (a\sqcup b)$.
\end{lemma}
\begin{proof} We have $ A a \sqcap  A  b =  A \neg a\sqcap  A \neg b\leq  A  (\neg a\sqcap\neg b)=  A  \neg (\neg a\sqcap\neg b)= A (a\sqcup b)$.
\end{proof}
 \noindent Of course, we could impose additional axioms, but our representation theorems will be stronger if we can represent \textit{any} epistemic awareness algebra, not just those with special additional properties.

Each epistemic possibility frame gives rise to an epistemic awareness algebra as follows.

\begin{restatable}{proposition}{FrameToAlg}\label{FrameToAlg} If $\mathscr{F}=(\Omega,\sqsubseteq, \mathcal{E}, \mathcal{A}, \mathcal{K},  \mathcal{B})$ is an epistemic possibility frame, then $\mathscr{F}^+=((\mathcal{E},\subseteq), \mathbf{A}, \mathbf{K}, \mathbf{B})$ is an epistemic awareness algebra.
\end{restatable}
\noindent An immediate corollary of  $\mathscr{F}^+$ satisfying tautology, symmetry, and agglomeration is that the family of events of which an agent is aware in a possibility forms a Boolean subalgebra of $\mathcal{RO}(\Omega,\sqsubseteq)$.
\begin{corollary}\label{SubalgCor} If $\mathscr{F}=(\Omega,\sqsubseteq, \mathcal{E}, \mathcal{A}, \mathcal{K}, \mathcal{B})$ is an epistemic possibility frame, then for any $\omega\in\Omega$, the family $\{E\in\mathcal{E}\mid \omega\in\mathbf{A} (E)\}$ contains $\Omega$ and is closed under $\cap$ and $\neg$ from (\ref{NegDef}).
\end{corollary}

We now proceed in the converse direction: given an epistemic awareness algebra, we wish to construct an epistemic possibility frame $\mathscr{F}$ that represents the algebra as $\mathscr{F}^+$. To this end, recall that a \textit{filter} in a Boolean algebra $\mathbb{B}$ is a nonempty set $F$ of elements of $\mathbb{B}$ that is upward closed under the lattice order of $\mathbb{B}$ (if $a\in F$ and $a\leq b$, then $b\in F$) and closed under the meet operation of $\mathbb{B}$ (if $a,b\in F$, then $a\sqcap b\in F$). A filter is \textit{proper} if it does not contain all elements of $\mathbb{B}$. Given an element $a\in \mathbb{B}$, the set $\mathord{\Uparrow}a=\{b\in B\mid a\leq b\}$ is a filter, and a filter $F$ is \textit{principal} if $F=\mathord{\Uparrow}a$ for some $a\in \mathbb{B}$. The following type of construction is used in possibility semantics (\citealt[\S~5.5]{Holliday2015}) but here we must add $\mathcal{A} $ for awareness and modify the treatment of $\mathcal{K} $ and $\mathcal{B} $ due to the awareness-restricted monotonicity of knowledge and belief.

\begin{restatable}{definition}{FilterConstruct}\label{FilterConstruct}Given an epistemic awareness algebra $\mathbb{A}=(\mathbb{B}, A, K, B)$, define   $\mathbb{A}_+=(\Omega,\sqsubseteq, \mathcal{E}, \mathcal{A}, \mathcal{K}, \mathcal{B})$ as follows:
\begin{enumerate}
\item $\Omega$ is the set of all proper filters of $\mathbb{B}$, and $F\sqsubseteq G$ if $F\supseteq G$;
\item $\mathcal{E}=\{\widehat{a}\mid a\in \mathbb{B}\}$ where $\widehat{a}=\{F\in \Omega\mid a\in F\}$;
\item $\mathcal{A} (F)=\{H\in \Omega\mid \mbox{$H$ is the principal filter of an element $a$ such that $ A a\in F$}\}$;
\item\label{FilterConstructK} $\mathcal{K} (F)=\{H\in\Omega\mid \mbox{for all }a_1,\dots,a_n\in\mathbb{B}\mbox{, if }K a_1\sqcup\dots\sqcup K a_n\in F\mbox{, then }a_1\sqcup\dots \sqcup a_n\in H\}$;
\item $\mathcal{B} (F)=\{H\in\Omega\mid \mbox{for all }a_1,\dots,a_n\in\mathbb{B}\mbox{, if }B a_1\sqcup\dots\sqcup B  a_n\in F\mbox{, then }a_1\sqcup\dots\sqcup a_n\in H\}$.
\end{enumerate}
\end{restatable}

In the Appendix, we prove the following main representation theorem.

\begin{restatable}{theorem}{FilterRep}\label{FilterRep} For any epistemic awareness algebra $\mathbb{A}$:
\begin{enumerate}
\item\label{FilterRep1} $\mathbb{A}_+$ is a standard epistemic possibility frame;
\item\label{FilterRep2} the map $a\mapsto\widehat{a}$ is an isomorphism from $\mathbb{A}$ to  $(\mathbb{A}_+)^+$. 
\end{enumerate}
\end{restatable}
\noindent This result is analogous to the Stone Representation Theorem for Boolean algebras (\citealt{Stone1936}) (or more precisely, the Stone-like representation without the Axiom of Choice\footnote{Using the Axiom of Choice, the set $\Omega$ in Definition \ref{FilterConstruct} may be cut down to just the set of proper filters that are maximal (i.e., ultrafilters) or principal.} in \citealt{Holliday2015} and \citealt{BH2020}) but with awareness, knowledge, and belief added alongside Boolean operations.

\begin{remark} In the case when $\mathbb{A}$ is finite, we can cut the set $\Omega$ in Definition \ref{FilterConstruct} down to just the principal filters. In fact, one can typically use a much smaller quasi-principal possibility frame, as in Examples \ref{GameEx} and \ref{OverconfidentEx}; one need only add enough partial possibilities to witness an agent's unawareness of events, rather than a partial possibility for each proposition from $\mathbb{A}$. For example, while the possibility frame in Example~\ref{GameEx} has 37 possibilities, its associated algebra has $2^{20}=1,048,576$ events. In practical modeling, we usually construct a concise possibility frame $\mathscr{F}$ and then calculate its algebra $\mathscr{F}^+$ of events, instead of starting with a very large algebra $\mathbb{A}$ of events and then applying Theorem \ref{FilterRep} to obtain a possibility frame $\mathbb{A}_+$. The point of Theorem \ref{FilterRep} is rather to show that possibility frames do not unreasonably limit what we can model.\end{remark}

\section{Conclusion}\label{Conclusion}

Theorem \ref{FilterRep} shows that epistemic possibility frames are capable of representing any scenario involving awareness of events, plus knowledge and belief---including that of multiple agents, given the obvious multi-agent generalization of everything above---provided some basic axioms are satisfied in the scenario. Thus, as far as event-based approaches to awareness are concerned, epistemic possibility frames provide a highly versatile modeling tool. We conclude by mentioning several avenues for further development.

First, we can immediately use our possibility frames to interpret a formal logical language for reasoning about unawareness and uncertainty, following standard practice in computer science (\citealt{Fagin1995}, \citealt{Halpern2003}) and some work in economics (e.g., \citealt{Board2004}, \citealt{Heifetz2008}, \citealt{Alon2014a}). We simply turn an epistemic possibility frame into an \textit{epistemic possibility model} for a propositional  language by  equipping the frame with a valuation $V$ of atomic formulas such that $V(\mathtt{p})\in \mathcal{E}$ for each atomic formula $\mathtt{p}$ of the language. One can then recursively define the interpretation of complex formulas built up using negation and  conjunction and modalities for awareness, knowledge, and belief for each agent, using the operations in the epistemic awareness algebra arising from the epistemic possibility frame. Finally, one can easily turn our main representation theorem (Theorem \ref{FilterRep}) into a strong completeness theorem for a  modal logic of awareness, knowledge, and belief with axioms matching those of epistemic awareness algebras.

A logical approach can then be extended to more expressive languages than that of propositional modal logic, such as propositional modal logic with \textit{propositional quantifiers}, as investigated in \citealt{Halpern2009,Halpern2013}. In such a language, one may express that an agent knows that there is some event of which she is unaware: $\mathtt{K} \exists\mathtt{p} \, \mathtt{U}  \mathtt{p}$. There is a mathematically elegant semantics for modal logic with propositional quantifiers using complete Boolean algebras (see, e.g., \citealt{Holliday2019}, \citealt{Ding2020}) and hence a corresponding semantics using possibility frames in which $\mathcal{E}=\mathcal{RO}(\Omega,\sqsubseteq)$ (\citealt[\S~5.1]{HollidayForthB}), which one could consider for applications to unawareness. Ding \citeyearpar{Ding2020} solves a related problem: modeling a \textit{modest} agent who believes that she must have some false belief, written as $\mathtt{B}  \exists \mathtt{p}(\mathtt{B} \mathtt{p}\wedge \neg\mathtt{p})$, though she does not know which belief it is.

There is also the possibility of using a formal language not just to reason about \textit{awareness of events} as modeled in this paper but also to model \textit{awareness of sentences}, perhaps even developing a sentence-based approach on top of our event-based approach. For example, assume that for every event $E\in \mathcal{E}$ in our model, there is some atomic formula $\mathtt{p}$ for which $V(\mathtt{p})=E$.  For an arbitrary formula $\varphi$ of the modal language, to formalize awareness of the formula $\varphi$ itself, we could say that $\mathtt{A} \varphi$ is true at a possibility $\omega$ just in case for every subformula $\psi$ of $\varphi$, we have $\omega\in \mathbf{A} \llbracket \psi\rrbracket$, where $\llbracket \psi\rrbracket$ is the set of possibilities at which $\psi$ is true (an event). This would capture the idea that being aware of complex formulas such as $\varphi\wedge\psi$ or $\varphi\vee\neg\varphi$ requires being aware of $\varphi$ and of $\psi$. However, this still has the consequence that if two atomic formulas $\mathtt{p}$ and $\mathtt{q}$ are true in exactly the same states, then $i$ is aware of $\mathtt{p}$ if and only if $i$ is aware of $\mathtt{q}$; to avoid this consequence, one could directly encode awareness of atomic formulas, as in \citealt{Fagin1988}, or one could associate with each possibility its own language and then require for the truth of $\mathtt{A} \varphi$ at $\omega$ that $\varphi$ belong to the languages associated with each possibility witnessing awareness of $\llbracket \varphi\rrbracket$ as in Definition \ref{PossFrameAware}.\ref{PossFrameAware3}.\footnote{Thanks to an anonymous reviewer for this latter suggestion.}
 
 Finally, though here we have focused on knowledge and qualitative belief, we can also add probability to our frames (cf.~\citealt{Aumann1999b}, \citealt{Heifetz2013}). This can be done by assigning to each possibility $\omega$ a set $\mathcal{P}_\omega$ of probability measures on the Boolean algebra $\{E\in \mathcal{E}\mid \omega\in \mathbf{A} (E)\}$ of events of which the agent is aware in $\omega$. The reason for allowing a \textit{set} of measures---besides wanting to allow multi-prior representations of uncertainty---is that a possibility $\omega$ may be \textit{partial}, not settling exactly which probability measure captures the agent's subjective probabilities, leaving us with a set of measures to be narrowed down by further refinements of $\omega$. Appropriate persistence and refinability conditions relating the sets $\mathcal{P}_\omega$ for different possibilities $\omega$  ensure that certain probabilistic events, such as the agent $p$-believing $E$ (\citealt{Monderer1989}) or judging that $E$ is at least as likely as $F$ (cf.~\citealt{Alon2014}, \citealt{Alon2014a}), will themselves belong to $\mathcal{RO}(\Omega,\sqsubseteq)$, so we can require that they belong to $\mathcal{E}$. Thus, a full apparatus of awareness, knowledge, and probability for multiple agents could be developed using possibility frames, enabling applications to decision theory and game theory with unawareness.

 \subsection*{Acknowledgements}
 
For helpful feedback, I thank the anonymous reviewers for \textit{Journal of Mathematical Economics}, as well as Mathias Boehm, Peter Fritz, John Hawthorne, Mikayla Kelley, Harvey Lederman, Matthew Mandelkern, Eric Pacuit, Burkhard Schipper, Snow Zhang, the audience in my course at NASSLLI 2022 at the University of Southern California, and the audience at my talk in the Zoom Workshop on Epistemic Logic with Unawareness in December 2022 organized by Burkhard Schipper.

\appendix

\section{The DLR impossibility theorem}\label{DLR}

\subsection{General formulation}\label{DLRgeneral}

Our goal in this paper is to model awareness of events as opposed to sentences. Famously, Dekel, Lipman, and Rustichini \citeyearpar{Dekel1998} prove theorems that are supposed to raise problems for such models. In this Appendix, we explain why our project is not doomed by these results. Although Dekel et al.~phrase their theorems in terms of state-space models, they are much more general. Let $\mathcal{E}$ be a nonempty set and $\leq$ a preorder on $\mathcal{E}$ with a minimum element $0$. Think of elements of $\mathcal{E}$ as events and $\leq$ as the entailment relation between events. For example,  given a nonempty set $\Omega$, we could take $\mathcal{E}$ to be the powerset of $\Omega$, $\leq \,=\,\subseteq$, and $0=\varnothing$, but this is just one example. Next, we assume maps  $U:\mathcal{E}\to\mathcal{E}$ for \textit{unawareness}, $K:\mathcal{E}\to\mathcal{E}$ for \textit{knowledge}, and $\neg:\mathcal{E}\to\mathcal{E}$ for \textit{negation}.\footnote{Stipulating some operations $U$ and $K$ on $\mathcal{E}$ is of course not to provide any illuminating model of unawareness and knowledge, but this abstract setup will be useful for stating impossibility theorems. The actual models of unawareness cited in this paper  (e.g., \citealt{Heifetz2006}, \citealt{Li2009}, \citealt{Fritz2015}, and our model) all attempt to \textit{represent} unawareness operations using more concrete structures. To provide any additional insight or representational succinctness beyond stipulating a primitive $U:\mathcal{E}\to\mathcal{E}$, the model must derive such an operation from more concrete relations, correspondences, etc.~on a set of states. For an example in which this requirement is not satisfied, note that in a state-space model based on a field of sets $(\Omega,\mathcal{E})$, stipulating an operation $U:\mathcal{E}\to\mathcal{E}$ is equivalent to stipulating a neighborhood function $N_U:\Omega\to\wp(\mathcal{E})$ via the definition: $E\in N_U(\omega)$ if and only if $\omega\in U(E)$. Thus, this repackaging with a neighborhood function (which merely lists the events of which an agent is supposed to be unaware at a state $\omega$) offers no additional insight or representational succinctness.} Suppose these maps satisfy the following axioms for all $E\in\mathcal{E}$: 
\begin{center}\vspace{-.15in}
\begin{minipage}{4in}
\begin{itemize}
\item $U(E)\leq U(U(E))$ (AU Introspection); 
\item $U(E)\leq \neg K(E)$ and $U(E)\leq \neg K\neg K(E)$ (Plausibility);
\item $K(U(E))=0$ (KU Introspection); 
\end{itemize}
\end{minipage}\begin{minipage}{3in}
\begin{itemize}
\item $K(\neg 0)=\neg 0$ (Necessitation);
\item $\neg\neg 0=0$ (Double Negation).
\end{itemize}
\end{minipage}
\end{center}
AU Introspection says that if an agent is unaware of $E$, then she is unaware that she is unaware of $E$, while KU introspection says that it is impossible for an agent to know that she is unaware of a specific event $E$. Necessitation says that the agent knows the trivial event $\neg 0$, namely $\Omega$ in state-space models, and Double Negation says that negation is involutive on $0$ as in Boolean algebras (and many generalizations thereof).

Plausibility is one direction of the Modica-Rustichini \citeyearpar{Modica1994} definition of unawareness in terms of knowledge. In fact, only the second half of Plausibility ($U(E)\leq \neg K\neg K(E)$) is used for the following result.
\begin{proposition}[\citealt{Dekel1998}, Theorem 1(i)]\label{Dekel1} Assuming the axioms above, $U(E)=0$ for all $E\in\mathcal{E}$.
\end{proposition}
\begin{proof} $U(E)\leq U(U(E))\leq \neg K\neg K(U(E))\leq \neg K \neg 0=\neg \neg 0=0$. 
\end{proof}
\noindent Thus, unawareness is contradictory and hence impossible assuming the axioms above.\footnote{Dekel et al.~prove another result (Theorem 1(ii)) that replaces Necessitation with the Monotonicity of $K$, i.e., if $E\leq F$, then $K(E)\leq K(F)$, but this principle is unacceptable assuming knowledge requires awareness, for reasons similar to those for rejecting monotonicity of awareness in Section \ref{Intro} (cf.~\citealt[p.~123]{Modica1994}, \citealt{Modica1999}).}  

Note, by contrast, that our argument concerning overconfidence at the beginning of Section \ref{Intro} targeted instead the converse of Plausibility, namely $\neg K(E)\sqcap \neg K\neg K(E)\leq U(E)$ (Converse Plausibility), where we now assume that for all $E,F\in\mathcal{E}$, there is a greatest lower bound $E\sqcap F$ of $\{E,F\}$ in $(\mathcal{E},\leq)$.

\begin{proposition} In addition to $K$, $U$, and $\neg$, assume a map $B:\mathcal{E}\to\mathcal{E}$ such that for all $E\in\mathcal{E}$, $B(E)\leq \neg K\neg E$ (Noncontradictory Belief and Knowledge). Further assume Converse Plausibility. Then for all $E\in\mathcal{E}$, we have (i) $B(K(E))\sqcap \neg K(E)\leq U(E)$. If in addition $B(E)\sqcap U(E)= 0 $ (Belief Requires Awareness), then (ii) $B(E) \sqcap B(K(E))\sqcap \neg K(E)=0$.
\end{proposition}
\begin{proof} For (i), $B(K(E))\leq \neg K\neg K(E)$, so $B(K(E))\sqcap \neg K(E)\leq \neg K(E)\sqcap \neg K\neg K(E)\leq U(E)$. Then for (ii), $B(E) \sqcap B(K(E))\sqcap \neg K(E)\leq B(E)\sqcap U(E)=0$.
\end{proof}
\noindent Thus, Converse Plausibility essentially renders overconfidence impossible, given that the other axioms are uncontroversial. We will also give an argument against the second half of Plausibility below.

Although Dekel et al.~originally framed their Theorem 1(i) as a result about standard state-space models, Proposition \ref{Dekel1} applies to any model of awareness that provides a structure $(\mathcal{E},\leq,0,U,K,\neg)$. For example, the model of Heifetz, Meier, and Schipper \citeyearpar{Heifetz2006} provides a set $\mathcal{E}$ constructed as certain pairs $(E,S)$ of sets ordered by $(E,S)\leq (E',S')$ if $E\subseteq E'$ and $S'\preceq S$, where $\preceq$ is a complete lattice order (see~\citealt[\S\S~2.1, 2.3, 2.5]{Schipper2013}). Hence the minimum element is $0=(\varnothing,S_{\top})$ where $S_\top$ is the maximum element of $\preceq$. The model also provides maps $U$, $K$, and $\neg$ from $\mathcal{E}$ to $\mathcal{E}$. Since this model allows for unawareness, it must reject one of Dekel et al.'s axioms. Indeed, it rejects $K(U(E))=0$ (KU Introspection). Heifetz et al.~instead set $K(U(E))$ to be a certain non-minimum element $(\varnothing,S(E))$ of $(\mathcal{E},\leq)$; although they denote that non-minimum element by $\emptyset^{S(E)}$ and refer to the axiom $K(U(E))=\emptyset^{S(E)}$ as ``KU Introspection,'' the event $\emptyset^{S(E)}$ is not equal to the minimum element 0 in  $(\mathcal{E},\leq)$, so their ``KU Introspection'' is not Dekel et al.'s axiom.\footnote{In response to this point, I have received the reply that the model of Heifetz et al.~preserves the spirit of Dekel et al.'s KU Introspection even though it does not preserve the mathematical letter of the axiom. However, the point that their model violates the axiom $K(U(E))=0$ is not merely a mathematical point; to instead set $K(U(E))=(\varnothing,S(E))$ is made possible by their abandoning the classical view that events form a Boolean algebra (see Remark \ref{AlgRemark}), which is conceptually significant.} The model of \citealt{Li2009} similarly rejects Dekel et al.'s KU Introspection axiom.

Like Heifetz et al.~and Li, Fritz and Lederman \citeyearpar[\S~2.3]{Fritz2015} respond to Dekel et al.'s result by rejecting their package of axioms. Taking inspiration from a distinction Dekel et al.~draw between ``real'' and ``subjective'' states in state-space models, Fritz and Lederman argue that AU Introspection, Plausibility, and KU Introspection need only hold with respect to some distinguished set $S$ of states, rather than with respect to the set $\Omega$ of all states. This means restricting the axioms as follows: $U(E)\sqcap S\leq U(U(E))$; $U(E)\sqcap S\leq \neg K(E)$ and $U(E)\sqcap S\leq \neg K\neg K(E)$;   and $K(U(E))\sqcap S=0$. Fritz and Lederman \citeyearpar[Theorem~2]{Fritz2015}  show that there are state-space models of unawareness of events satisfying these restricted axioms, plus Necessitation and Double Negation, in which unawareness is possible. More recently, Fukuda \citeyearpar{Fukuda2021} has studied other models violating AU Introspection in particular.

Thus, there is precedent in the literature for rejecting one or more of the axioms of Dekel et al. In the rest of this appendix, we will raise doubts about their axioms in light of the distinction between awareness of events vs.~awareness of sentences.\footnote{After writing this paper, I learned from Harvey Lederman that \citealt{Elliot2022} raises similar doubts about the axioms.}

\subsection{Plausibility}\label{PlausibilityAppendix}

In particular, we will raise doubt about the second half of Plausibility: $U(E)\leq \neg K\neg K(E)$. Suppose $F$ is some in principle \textit{unknowable} event, meaning $K(F)=0$, such that our agent is also unaware of $F$. According to KU Introspection, $U(E)$ is an example of such an unknowable $F$, but we need not use KU Introspection; under highly plausible assumptions, $F:=E\sqcap \neg U(E)\sqcap \neg K(E)$ is unknowable\footnote{$E\sqcap \neg K(E)$ is the classic example of an unknowable event from what is known as Fitch's paradox (\citealt{Fitch1963}); for example, can Ann know the event expressed by ``Bob played left but Ann doesn't know it''? We could use this simpler event and the axiom $K(E\sqcap \neg K(E))=0$, but we will instead derive $K(E\sqcap \neg U(E)\sqcap \neg K(E))=0$ from other axioms.}  (see the proof of Proposition \ref{FitchProp}). Now since $F$ is unknowable,  $\neg K(F)$ is equal to $\Omega$ (or abstractly, $\neg 0$). But any agent knows $\Omega$ (again cf.~\citealt{Stalnaker1984}), being careful not to confuse $\Omega$ with any particular sentence true throughout $\Omega$. Thus, the agent's unawareness of $F$ contradicts Plausibility, which implausibly implies that an agent who knows $\Omega$ must be aware of every unknowable event. It is surely true that knowledge of a \textit{sentence} $\neg\mathtt{K}\varphi$ implies awareness of the \textit{embedded sentence} $\varphi$. But unlike sentences, events do not intrinsically ``embed'' other events; and there is no reason to suppose that an agent who knows $\Omega$ thinks about the event $\Omega$ as the event expressed by the sentence $\neg\mathtt{K}\varphi$.\footnote{Where $A(E)=\neg U(E)$, these points also show the problem with the principle $A(\neg K(E))\leq A(E)$ when $\neg K(E)$ is $\Omega$. Note that this principle follows from $ A (K(E))\leq A(E)$ (one direction of AK-Self Reflection in \citealt[Prop.~3]{Heifetz2006}) and $A(F)=A(\neg F)$ (symmetry), casting doubt on the former when $K(E)=0$. See Footnote \ref{OtherPrinciples} for a restricted principle.} See Example \ref{WatsonKnowledge} for a similar point. Now let us prove the claim above.  

\begin{proposition}\label{FitchProp} Assume Necessitation and Double Negation as above. Further suppose that for all $E,F\in\mathcal{E}$, we have $E\sqcap \neg E=0$ (Noncontradiction), $K(E)\leq E$ (Factivity), and that if $E\leq F$, then $K(E)\sqcap \neg U(F)\leq K(F)$ (Awareness-restricted Monotonicity, which appears in Section \ref{Representation}). Finally, assume the second half of Plausibility. Then for every $E\in \mathcal{E}$, we have $U(E\sqcap \neg U(E)\sqcap\neg K(E))=0$.
\end{proposition}
\begin{proof} Where $E'= E\sqcap \neg U(E)\sqcap \neg K(E)$, we have $K(E')\leq E'\leq \neg U(E)$ by Factivity, which implies $K(E')\leq K(E')\sqcap \neg U(E)\leq K (E)$ using Awareness-restricted Monotonicity with  $E'\leq E$. But we also have $K(E')\leq E'\leq \neg K(E)$ using Factivity. Thus, $K(E')\leq K(E)\sqcap\neg K(E)=0$, so $K(E')=0$. Hence by Necessitation, $K\neg K(E')=\neg 0$. Then by the second half of Plausibility, $U(E')\leq \neg K\neg K(E')= \neg\neg 0=0$.\end{proof}
\noindent For $E'$ as in the proof, we should be able to have $E'\neq 0$ and $E' \neq \neg 0$, so the idea that it is impossible to be unaware of $E'$ is highly counterintuitive; indeed, Example \ref{OverconfidentEx} shows how such unawareness is possible.\footnote{\label{FitchNote}In Example \ref{OverconfidentEx}, if we define $Fraud'=Fraud\cap \mathbf{A} (Fraud)\cap \neg\mathbf{K} (Fraud)$, then $Fraud'=\{f_2,f_4\}$ and $U(Fraud')$ is the set of blue and green states.}  We take Proposition \ref{FitchProp}  to cast doubt on the second half of Plausibility, rather than any of the other axioms. Fortunately, Fact \ref{PossibilityFact} shows that a slight weakening of Plausibility delivers us away from Dekel et al.'s~\citeyearpar{Dekel1998} impossibility theorem to a possibility result.

Finally, we address whether the problems with the Modica-Rustichini definition of unawareness in terms of knowledge might be solved by changing the definition to use belief: $U(E)=\neg B(E)\sqcap \neg B\neg B(E)$. The answer is that under a standard assumption of  models of rational belief, namely that  $B(E\sqcap \neg B(E))=0$ (No Moorean Beliefs\footnote{The name is taken from what is known as Moore's paradox (see \citealt[\S~4.5]{Hintikka2005} and \citealt{Holliday2010}). The standard axiomatization of the   \textit{implicit belief} operator (not requiring awareness), denoted by $L$ in \citealt{Fagin1988}, entails No Moorean Belief for $L$ in place of $B$. Then given that explicit belief (requiring awareness) entails implicit belief, we can derive $B(E\sqcap\neg L(E))=0$ and replace the conclusion of Proposition \ref{MooreProp} with the equally unappealing $U(E\sqcap \neg L(E))=0$.}), the same kind of argument given above against the second half of Plausibility can be applied to the second half of Belief Plausibility, i.e., to $U(E)\leq \neg B\neg B(E)$.\footnote{Note that the belief modification of the Modica-Rustichini definition assumes an agent who is not mistaken about what she believes ($B(B(E))\leq B(E)$). For otherwise we could have an agent who is aware of $E$, does not believe $E$, but believes that she does believe $E$, so she does not believe that she does not believe $E$, contradicting the modified definition. From here it is a short step to the assumption that the agent is not mistaken about what she does not believe and then to No Moorean Beliefs.}

\begin{proposition}\label{MooreProp} Assume Double Negation, Necessitation for $B$ instead of $K$, and No Moorean Beliefs. In addition, assume the second half of Belief Plausibility. Then for all events $E\in\mathcal{E}$, we have  $U(E\sqcap \neg B(E))=0$.
\end{proposition}
\begin{proof} Where $E'=E\sqcap \neg B(E)$, we have $B(E')=0$ by No Moorean Beliefs, so $B\neg B(E')=\neg 0$ by Necessitation, in which case the second half of Belief Plausibility yields $U(E')\leq \neg B\neg B(E')= \neg\neg 0=0$.\end{proof}

\noindent In short, Belief Plausibility implausibly implies that an agent who believes $\Omega$ must be aware of every \textit{unbelievable} event $E'$ as in the proof. To see the implausibility, note that many people at the beginning of World War II were unaware of the unbelievable $E'$: a nuclear weapon will end the war, but we do not believe that a nuclear weapon will end the war. Thus, changing the Modica-Rustichini definition of unawareness to use belief instead of knowledge still does not provide a satisfactory definition of unawareness of events.

\section{Proofs}\label{Proofs}

In this appendix, we give proofs of results in the main text. To make clear where we use our various assumptions, we typeset them in bold. 

\FiniteJoin*
\begin{proof} First, assume (i). We prove (ii) by induction on $n$. For the base case of $n=1$, if $\nu_1\in \mathcal{A} (\omega)\cap \mathord{\downarrow}\nu$, so $\nu_1\sqsubseteq\nu$, then $\mathrm{max}(\mathord{\downarrow}\nu_1\cap \mathord{\downarrow}\nu)=\{\nu_1\}\subseteq \mathcal{A} (\omega)$. For the inductive step, suppose that $\nu\in \mathcal{A} (\omega)$ and $\nu_1,\dots,\nu_{n+1}\in \mathcal{A} (\omega)\cap \mathord{\downarrow}\nu$.  By the inductive hypothesis, we have $\mathrm{max}((\mathord{\downarrow}\nu_1\sqcup\dots\sqcup\mathord{\downarrow}\nu_n)\cap\mathord{\downarrow}\nu)\subseteq \mathcal{A} (\omega)$; moreover, we have $\mathrm{max}(\mathord{\downarrow}\nu_{n+1} \cap \mathord{\downarrow}\nu )=\{\nu_{n+1}\}\subseteq\mathcal{A} (\omega)$. By \textbf{awareness expressibility},   $\mathord{\downarrow}\nu_i \in\mathcal{E}$ for $1\leq i\leq n+1$, so $\mathord{\downarrow}\nu_1\sqcup\dots\sqcup\mathord{\downarrow}\nu_n \in\mathcal{E}$ by Definition \ref{PossFrames}, and $\mathord{\downarrow}\nu_{n+1}\in \mathcal{E}$. By the previous two sentences and \textbf{awareness joinability}, we have $\mathrm{max}( (\mathord{\downarrow}\nu_1\sqcup\dots\sqcup\mathord{\downarrow}\nu_n \sqcup \mathord{\downarrow}\nu_{n+1})\cap\mathord{\downarrow}\nu)\subseteq\mathcal{A} (\omega)$, which establishes (ii) for $n+1$.

Now assume (ii) and that $\mathrm{max}(E)$ is finite for each $E\in\mathcal{E}$. Toward proving (i), suppose $\nu \in \mathcal{A} (\omega)$, $E,E'\in\mathcal{E}$, and that $\mathrm{max}(E\cap \mathord{\downarrow}\nu)\cup \mathrm{max}(E'\cap \mathord{\downarrow}\nu)\subseteq \mathcal{A} (\omega)$.  By the finiteness assumption, we can write $\mathrm{max}(E\cap \mathord{\downarrow}\nu)\cup \mathrm{max}(E'\cap \mathord{\downarrow}\nu)=\{\nu_1,\dots,\nu_n\}$. Then by (ii), $\mathrm{max}((\mathord{\downarrow}\nu_1\sqcup\dots\sqcup\mathord{\downarrow}\nu_n)\cap\mathord{\downarrow}\nu)\subseteq\mathcal{A} (\omega)$. Now we claim that $(\mathord{\downarrow}\nu_1\sqcup\dots\sqcup\mathord{\downarrow}\nu_n)\cap\mathord{\downarrow}\nu= (E\sqcup E')\cap \mathord{\downarrow}\nu$. For the left-to-right inclusion, since $\nu_1,\dots,\nu_n\in E\cup E'$ and $E$ and $E'$ are downsets, we have $\mathord{\downarrow}\nu_1\cup\dots\cup\mathord{\downarrow}\nu_n\subseteq E\cup E'$, so $\rho(\mathord{\downarrow}\nu_1\cup\dots\cup\mathord{\downarrow}\nu_n)\subseteq \rho(E\cup E')$ by Lemma \ref{RhoLem}.1 and hence  $\mathord{\downarrow}\nu_1\sqcup\dots\sqcup\mathord{\downarrow}\nu_n\subseteq E\sqcup E'$ by Theorem \ref{Tarski2}. For the right-to-left inclusion, suppose $\mu\in (E\sqcup E')\cap \mathord{\downarrow}\nu$. To show $\mu\in\mathord{\downarrow}\nu_1\sqcup\dots\sqcup\mathord{\downarrow}\nu_n$, it suffices to show that for every $\mu'\sqsubseteq\mu$, there is a $\mu''\sqsubseteq\mu'$ with $\mu''\in \mathord{\downarrow}\nu_1\cup\dots\cup\mathord{\downarrow}\nu_n$. Given $\mu'\sqsubseteq\mu$ and $\mu\in E\sqcup E'$, there is a $\mu''\sqsubseteq\mu'$ with $\mu''\in E\cup E'$. Without loss of generality, suppose $\mu''\in E$. Then since $\mu''\sqsubseteq \mu'\sqsubseteq\mu\sqsubseteq\nu$, we have $\mu''\in E\cap \mathord{\downarrow}\nu$. By \textbf{awareness expressibility}, $\mathord{\downarrow}\nu\in \mathcal{E}$, so we have $E\cap \mathord{\downarrow}\nu\in\mathcal{E}$. Then since $(\Omega,\sqsubseteq,\mathcal{E})$ is \textbf{quasi-principal} and $\mu''\in E\cap \mathord{\downarrow}\nu$, there is a $\nu_i\in \mathrm{max}(E\cap \mathord{\downarrow}\nu)$ such that $\mu''\sqsubseteq \mu_i$. Hence $\mu''\in \mathord{\downarrow}\nu_1\cup\dots\cup\mathord{\downarrow}\nu_n$, which completes the proof of (i).\end{proof}

Next we prove the two key lemmas that together show that the set of regular open sets of an epistemic possibility frame is closed under $\mathbf{A} $, $\mathbf{K} $, and $\mathbf{B} $.

\ROclosure*

\begin{proof} By Lemma \ref{ROLem}, it suffices to verify persistence and refinability for $\mathbf{A} (E)$. Persistence is immediate from the $\forall \omega'\sqsubseteq \omega$ quantification in the definition of $\mathbf{A} $. As for refinability, suppose $\omega\not\in\mathbf{A}  (E)$, so there is some $\omega'\sqsubseteq \omega$, $\nu \in  \mathcal{A} (\omega')$, and $\nu'\in\mathrm{max}(E\cap\mathord{\downarrow}\nu)\cup \mathrm{max}(\neg E\cap\mathord{\downarrow}\nu)$ such that $\nu'\not\in  \mathcal{A} (\omega')$. Given $\nu'\not\in  \mathcal{A} (\omega')$, by \textbf{awareness refinability} there is an $ \omega''\sqsubseteq \omega'$ such that for all $ \omega'''\sqsubseteq \omega''$, we have $\nu'\not\in \mathcal{A} (\omega''')$.  We claim that for all $\omega'''\sqsubseteq \omega''$, we have $\omega'''\not\in \mathbf{A} (E)$. Assume for contradiction that $\omega'''\in \mathbf{A} (E)$. Since $\omega'''\sqsubseteq \omega''\sqsubseteq \omega'$, we have $\omega'''\sqsubseteq \omega'$, which with $\nu \in  \mathcal{A} (\omega')$ implies $\nu \in  \mathcal{A} (\omega''')$ by \textbf{awareness persistence}. Together $\omega'''\in \mathbf{A}  (E)$ and $\nu \in  \mathcal{A} (\omega''')$ imply $\mathrm{max}(E\cap\mathord{\downarrow}\nu)\cup \mathrm{max}(\neg E\cap\mathord{\downarrow}\nu) \subseteq  \mathcal{A} (\omega''')$. Hence $\nu'\in  \mathcal{A} (\omega''')$, contradicting what we derived above. Thus, $\omega'''\not\in \mathbf{A} (E)$, which establishes refinability for $\mathbf{A}  (E)$.\end{proof}

\ROclosureTwo*

\begin{proof} By Lemma \ref{ROLem}, it suffices to verify persistence and refinability for $\{\omega\in\Omega\mid \mathcal{R} (\omega)\subseteq E\}$.  For persistence, suppose $\omega'\sqsubseteq\omega$. Then by \textbf{$\mathcal{R} $-monotonicity}, $\mathcal{R} (\omega')\subseteq\mathcal{R} (\omega)$, so  $\mathcal{R} (\omega)\subseteq E$ implies ${\mathcal{R} (\omega')\subseteq E}$. For refinability, suppose $\mathcal{R} (\omega)\not\subseteq E$, so there is some $\nu\in \mathcal{R} (\omega)\setminus E$.  Since $\nu\not\in E$ and $E\in\mathcal{RO}(\Omega,\sqsubseteq)$, there is a $\nu'\sqsubseteq \nu$ such that for all $\nu''\sqsubseteq \nu'$, $\nu''\not\in E$. Since $\nu\in \mathcal{R} (\omega)$ and $\nu'\sqsubseteq\nu$, we have $\nu'\in\mathcal{R} (\omega)$ by \textbf{$\mathcal{R} $-regularity}. Then by  \textbf{$\mathcal{R} $-refinability}, there is an $ \omega'\sqsubseteq \omega$ such that for all $ \omega''\sqsubseteq\omega'$ there is a $ \nu''\sqsubseteq \nu'$ with $\nu''\in \mathcal{R} (\omega'')$. But as above, $\nu''\sqsubseteq\nu' $ implies $\nu''\not\in E$, so $\mathcal{R} (\omega'')\not\subseteq E$. Thus, we have shown that for all $\omega'\sqsubseteq \omega$ there is an $\omega''\sqsubseteq\omega'$ with $\mathcal{R} (\omega'')\not\subseteq E$, which completes the proof of refinability.\end{proof}

\noindent By these lemmas and the fact that $\mathcal{RO}(\Omega,\sqsubseteq)$ is closed under intersection, it is also closed under $\mathbf{K} $ and $\mathbf{B}$.

Next we prove that epistemic possibility frames as in Definition \ref{EpistemicFrame} give rise to epistemic awareness algebras as in Definition \ref{AwareAlg}.

\FrameToAlg*

\begin{proof} That $\mathbf{A} $ satisfies the tautology and symmetry axioms of Definition \ref{AwareAlg} is obvious. For the agglomeration axiom, suppose $\omega\in \mathbf{A} (E_1)\cap \mathbf{A} (E_2)$. Toward showing that $\omega\in \mathbf{A} (E_1\cap E_2)$, suppose $\omega'\sqsubseteq \omega$ and   $\nu\in\mathcal{A} (\omega')$, so $\mathord{\downarrow}\nu\in \mathcal{E}$ by \textbf{awareness expressibility}. Since $E_1,E_2\in \mathcal{E}$, we have $E_1\cap \mathord{\downarrow}\nu, E_2\cap\mathord{\downarrow}\nu\in \mathcal{E}$, as well as $\neg E_1\cap \mathord{\downarrow}\nu, \neg E_2\cap \mathord{\downarrow}\nu\in\mathcal{E}$. We must show that $ \mathrm{max}(E_1\cap E_2\cap \mathord{\downarrow} \nu)\cup \mathrm{max}(\neg ( E_1\cap  E_2)\cap \mathord{\downarrow} \nu)\subseteq\mathcal{A}(\omega')$.

First suppose that $\nu'\in \mathrm{max}(E_1\cap E_2\cap \mathord{\downarrow} \nu)$. Then since $\nu'\in E_1\cap \mathord{\downarrow} \nu\in\mathcal{E}$ and the frame is \textbf{quasi-principal} (recall Definition \ref{QP}), there is some $\nu^*\in \mathrm{max}(  E_1\cap\mathord{\downarrow} \nu)$ with $\nu'\sqsubseteq\nu^*$. Then since $\omega\in \mathbf{A}  (E_1)$, $\omega'\sqsubseteq \omega$, $\nu\in \mathcal{A} (\omega')$, and $\nu^*\in  \mathrm{max}( E_1\cap\mathord{\downarrow} \nu)$,  we have $\nu^*\in \mathcal{A} (\omega')$. Moreover, we have not only $\nu'\in E_2\cap \mathord{\downarrow}\nu^*$ from above  but also  $\nu'\in \mathrm{max}( E_2\cap \mathord{\downarrow}\nu^*)$; for if there is some $\nu''$ such that $\nu'\sqsubset \nu'' \in  E_2\cap \mathord{\downarrow}\nu^*$, then given that $\nu^*\in  \mathrm{max}( E_1\cap\mathord{\downarrow} \nu)$ and that $E_1$ is a downset, we have $\nu''\in  E_1\cap  E_2\cap \mathord{\downarrow}\nu$, contradicting the fact that $\nu'\in \mathrm{max}( E_1\cap  E_2\cap \mathord{\downarrow} \nu)$. Then since $\omega\in \mathbf{A}  (E_2)$, $\omega'\sqsubseteq \omega$, $\nu^*\in \mathcal{A} (\omega')$, and $\nu'\in \mathrm{max}( E_2\cap \mathord{\downarrow}\nu^*)$, we have~$\nu'\in \mathcal{A} (\omega')$. Thus, $\mathrm{max}(E_1\cap E_2\cap \mathord{\downarrow} \nu)\subseteq \mathcal{A}(\omega')$.

Now we show $\mathrm{max}(\neg ( E_1\cap  E_2)\cap \mathord{\downarrow} \nu)\subseteq\mathcal{A}(\omega')$. Since $\omega\in \mathbf{A}  (E_1)\cap \mathbf{A}  (E_2)$, $\omega'\sqsubseteq \omega$, and   $\nu\in\mathcal{A} (\omega')$, we have $\mathrm{max}(\neg  E_1\cap \mathord{\downarrow} \nu)\cup \mathrm{max}(\neg  E_2\cap \mathord{\downarrow} \nu)\subseteq  \mathcal{A} (\omega')$.  Thus, by \textbf{awareness joinability}, $\mathrm{max}((\neg E_1\sqcup \neg E_2)\cap \mathord{\downarrow}\nu)\subseteq \mathcal{A}(\omega')$, which with $\neg E_1\sqcup \neg E_2=\neg (E_1\cap E_2)$ yields the desired inclusion. This completes the proof of $\omega\in \mathbf{A} (E_1\cap E_2)$.

Checking the axioms on knowledge and belief from Definition \ref{AwareAlg} is straightforward\end{proof}

Finally, we prove our main representation theorem. For convenience, we repeat the definition of $\mathbb{A}_+$.

\FilterConstruct*

\FilterRep*

\begin{proof} For part \ref{FilterRep1}, let $\mathbb{B}$ be the underlying Boolean algebra of $\mathbb{A}$. Note that for a filter $F$ in $\mathbb{B}$ and $b\in \mathbb{B}$, the smallest filter extending $F\cup\{b\}$ is $\{c\in \mathbb{B}\mid \mbox{for some }a\in F, a\sqcap b\leq c\}$. A basic fact about Boolean algebras is that if $F$ is a proper filter and $b\not\in F$, then the smallest filter extending $F\cup\{\neg b\}$ is proper.

Clearly $(\Omega,\sqsubseteq)$ is a poset. To see that $(\Omega,\sqsubseteq, \mathcal{E})$ is a possibility frame, we must show that $\mathcal{E}$ is a subalgebra of $\mathcal{RO}(\Omega,\sqsubseteq)$. First, we show that each $\widehat{a}$ is regular open. By Lemma \ref{ROLem}, it suffices to show that $\widehat{a}$ satisfies persistence and refinability. For persistence, if $F\in \widehat{a}$, so $a\in F$, and $F'\sqsubseteq F$, so $F'\supseteq F$, then $a\in F'$ and hence $F'\in\widehat{a}$. For refinability, if $F\not\in\widehat{a}$, so $a\not\in F$, then the smallest filter $F'$ extending $F\cup\{\neg a\}$ is proper, so $F'\sqsubseteq F$, and for all $F''\sqsubseteq F'$, we have $\neg a\in F''$, so $a\not\in F''$ since $F''$ is proper, so $F''\not\in\widehat{a}$. Finally,  $\mathcal{E}$ is closed under $\cap$ and $\neg$, as (i) $\widehat{a}\cap\widehat{b}= \widehat{a\sqcap b}$, and (ii) $\neg\widehat{a} = \widehat{\neg a}$. Condition (i) follows from that fact that if $F$ is a filter, then $a,b\in F$ if and only if $a\sqcap b\in F$. For condition (ii), if $F\in \widehat{\neg a}$, so $\neg a\in F$, then for any proper filter $F'$ extending $F$, $a\not\in F'$, so $F'\not\in\widehat{a}$; this shows $F\in \neg\widehat{a}$. Conversely, if $F\not\in\widehat{\neg a}$, so $\neg a\not\in F$, then the smallest filter $F'$ extending $F\cup \{a\}$ is proper, so $F'\sqsubseteq F$; this shows $F\not\in\neg\widehat{a}$.

To show that $(\Omega,\sqsubseteq)$ is quasi-principal, we must show that for all $\widehat{a}\in \mathcal{E}$ and $F\in \widehat{a}$, we have that $F\in \mathord{\downarrow}\mathrm{max}(\widehat{a})$. Indeed, $\mathrm{max}(\widehat{a})=\{\mathord{\Uparrow} a\}$ (recall that $\mathord{\Uparrow} a$ is the principal filter generated by $a$), and from $F\in\widehat{a}$ we have $a\in F$, so $F\supseteq \mathord{\Uparrow} a$ and hence $F\sqsubseteq \mathord{\Uparrow}a$, so $F\in \mathord{\downarrow} \mathrm{max}(\widehat{a})$. Finally, the maximum element $\maximum $ of the poset of proper filters is the principal filter $\{1\}$.

Next we verify the five conditions on $\mathcal{A} $:
\begin{itemize}
\item awareness nonvacuity: for all $F\in \Omega$, $\maximum \in \mathcal{A} (F)$.

We have $\{1\}\in\mathcal{A} (F)$ by the \textbf{tautology} axiom of Definition \ref{AwareAlg}.

\item awareness expressibility: for all $F\in \Omega$ and $H\in \mathcal{A} (F)$, $\mathord{\downarrow}H\in \mathcal{E}$.

If $H\in \mathcal{A} (F)$, then $H=\mathord{\Uparrow}a$ for some $a\in \mathbb{B}$, in which case $\mathord{\downarrow}H=\widehat{a}\in \mathcal{E}$.

\item awareness persistence: if $F'\sqsubseteq F$, then $\mathcal{A} (F)\subseteq \mathcal{A} (F')$.

Immediate from the definitions of $\mathcal{A} $ and $\sqsubseteq$ in $\mathbb{A}_+$.

\item awareness refinability: if  $H\not\in \mathcal{A}  (F)$, then $\exists F'\sqsubseteq F$ $\forall F''\sqsubseteq F'$  $H\not\in \mathcal{A} (F'' )$.

Suppose $H\not\in \mathcal{A}  (F)$. If $H$ is not a principal proper filter, then set $F'=F$, and then for all $F''\sqsubseteq F'$, we have $H\not\in \mathcal{A}  (F'')$. Now suppose $H$ is the principal filter of a nonzero element $a$. Since $H\not\in  \mathcal{A}  (F)$, it follows that $ A  a\not\in F$. Hence the smallest filter $F'$ extending $F\cup \{\neg A  a\}$ is proper. Then for all $F''\sqsubseteq F'$, i.e., $F''\supseteq F'$, we have $\neg A  a\in F''$ and hence $ A  a\not\in F''$ since $F''$ is proper, so $H\not\in \mathcal{A}  (F'' )$.

\item awareness joinability: if $H\in \mathcal{A}  (F)$ and $\mathrm{max}(\widehat{a}\cap \mathord{\downarrow}H)\cup \mathrm{max}(\widehat{b}\cap \mathord{\downarrow}H)\subseteq  \mathcal{A} (F)$, then \[\mathrm{max}((\widehat{a}\sqcup\widehat{b})\cap\mathord{\downarrow}H)\subseteq  \mathcal{A} (F).\]
Assume the antecedent of the above conditional. Given $H\in  \mathcal{A}  (F)$, we have $H=\mathord{\Uparrow}c$ for some $c\in\mathbb{B}$. If $(a\sqcup b)\sqcap c= 0$, then $\mathrm{max}((\widehat{a}\sqcup\widehat{b})\cap\mathord{\downarrow}H)=\mathrm{max}(\varnothing)=\varnothing\subseteq \mathcal{A}(F)$. If $(a\sqcup b)\sqcap c\neq 0$, then $\mathrm{max}((\widehat{a}\sqcup\widehat{b})\cap\mathord{\downarrow}H)=\{\mathord{\Uparrow}((a\sqcup b)\sqcap c)\} = \{\mathrm{\Uparrow} ((a\sqcap c)\sqcup (b\sqcap c))\}$ by the distributive law of Boolean algebras. Now for $d\in \{a,b\}$, either $d\sqcap c=0$, in which case $A(d\sqcap c)\in F$ by the \textbf{tautology} and \textbf{symmetry} axioms, or $d\sqcap c\neq 0$, in which case our assumption that  $\mathrm{max}(\widehat{d}\cap \mathord{\downarrow}H)\subseteq  \mathcal{A} (F)$ implies $A  (d\sqcap c)\in F$. Hence $A ((a\sqcap c)\sqcup (b\sqcap c))\in F$ by Lemma \ref{JoinClosure}, so $\mathord{\Uparrow}((a\sqcap c)\sqcup (b\sqcap c))\in\mathcal{A}  (F)$. Combining this with what we derived above using distributivity, we conclude that $\mathrm{max}((\widehat{a}\sqcup\widehat{b})\cap\mathord{\downarrow}H)\subseteq  \mathcal{A} (F)$.
 
\end{itemize}
Finally, we verify the conditions on $\mathcal{R} \in \{\mathcal{K} ,\mathcal{B} \}$, reasoning about $\Box \in\{K ,B \}$:
\begin{itemize}
\item $\mathcal{R} $-monotonicity: if $F'\sqsubseteq F$, then  $\mathcal{R} (F')\subseteq \mathcal{R} (F)$.

Assume $F'\sqsubseteq F$ and $H\in \mathcal{R} (F')$. Toward showing $H\in\mathcal{R} (F)$, suppose  $\Box  a_1\sqcup\dots\sqcup \Box  a_n\in F$. Then since $F'\sqsubseteq F$, $F'\supseteq F$, so $\Box  a_1\sqcup\dots\sqcup \Box  a_n\in F'$, which with $H\in \mathcal{R} (F')$ implies $a_1\sqcup\dots\sqcup a_n\in H$.

\item $\mathcal{R} $-regularity: $\mathcal{R} (F)\in\mathcal{RO}(\Omega,\sqsubseteq)$. 

By Lemma \ref{RhoLem}, it suffices to show that $\mathcal{R} (F)$ satisfies persistence and refinability. That  $\mathcal{R} (F)$ satisfies persistence is immediate from the definition of $\mathcal{R} $ and the definition of $\sqsubseteq$ as $\supseteq$. For refinability, suppose $H\not\in\mathcal{R} (F)$. Hence there is $\Box  a_1\sqcup\dots\sqcup \Box a_n\in F$ such that $a_1\sqcup\dots\sqcup a_n\not\in H$. It follows that the smallest filter $H'$ extending $H\cup \{\neg (a_1\sqcup\dots\sqcup a_n)\}$ is proper. Now suppose $H''\sqsubseteq H'$, so $H''$ is a proper filter with $H''\supseteq H'$ and hence $\neg (a_1\sqcup\dots\sqcup a_n)\in H''$. Then $a_1\sqcup\dots\sqcup a_n\not\in H''$, so $H''\not\in \mathcal{R} (F)$. 

\item $\mathcal{R} $-refinability: if $H \in \mathcal{R} (F)$, then $\exists F'\sqsubseteq F$ $\forall F''\sqsubseteq F'$ $\exists H'\sqsubseteq H$: $H'\in \mathcal{R} (F'')$. 

Assuming $H\in \mathcal{R} (F)$, we claim that the smallest filter $F'$ extending \[F\cup\{ \neg \Box  c \mid \exists b\in H:c\leq\neg b\}\] is proper. If not, then there are $a\in F$, $b_1,\dots,b_n\in H$, and $c_1,\dots,c_n$ with $c_k\leq \neg b_k$ for $1\leq k\leq n$ such that $a\sqcap \neg\Box  c_1 \sqcap\dots\sqcap\neg\Box  c_n=0$, which implies  $a\leq \neg ( \neg\Box  c_1\sqcap\dots\sqcap\neg\Box  c_n)= \Box  c_1\sqcup\dots\sqcup \Box  c_n$. Hence $\Box  c_1\sqcup\dots\sqcup \Box  c_n\in F$, which with $H\in \mathcal{R} (F)$ implies $c_1\sqcup\dots\sqcup c_n\in H$ and hence $\neg b_1\sqcup\dots\sqcup \neg b_n\in H$, contradicting the fact that $b_1,\dots,b_n\in H$ and $H$ is a proper filter. Thus, $F'$ is indeed proper, so $F'\sqsubseteq F$. Now consider any $F''\sqsubseteq F'$, so $F''\supseteq F'$. We claim that the smallest filter $H'$ extending \[H\cup \{a_1\sqcup\dots\sqcup a_n \mid \Box a_1\sqcup\dots\sqcup \Box a_n\in F''\}\] is proper. If not, there is some $b\in H$ and a family $\{a_1^j\sqcup\dots\sqcup a_{n_j}^j\}_{j\in J}$ for a nonempty finite $J$ such that $\Box a_1^j\sqcup\dots\sqcup \Box a_{n_j}^j\in F''$ and
\[ \underset{j\in J}{\bigsqcap} (a_1^j\sqcup\dots\sqcup a_{n_j}^j)\leq \neg b.\]
Let $\Lambda$ be the set of all choice functions $\lambda$ such that for $j\in J$,  $1\leq \lambda(j)\leq n_j$. Then for each $\lambda\in \Lambda$,
\[\underset{j\in J}{\bigsqcap} a_{\lambda(j)}^j\leq\neg b. \]
Hence for each $\lambda\in \Lambda$, we have $\neg \Box  \underset{j\in J}{\bigsqcap} a_{\lambda(j)}^j\in F'$ by construction of $F'$ and hence $\neg \Box  \underset{j\in J}{\bigsqcap} a_{\lambda(j)}^j\in F''$. On the other hand, 
\begin{equation} \underset{j\in J}{\bigsqcap} (\Box  a_1^j\sqcup\dots\sqcup \Box  a_{n_j}^j)\leq \underset{\lambda\in \Lambda}{\bigsqcup}\, \underset{j\in J}{\bigsqcap} \Box  a_{\lambda(j)}^j \leq \underset{\lambda\in \Lambda}{\bigsqcup} \Box \underset{j\in J}{\bigsqcap}  a_{\lambda(j)}^j,\label{DistAglom}\end{equation}
using Boolean distributivity for the first inequality and \textbf{knowledge/belief agglomeration} for the second, since $J$ is finite. Then since $\underset{j\in J}{\bigsqcap} (\Box  a_1^j\sqcup\dots\sqcup \Box  a_{n_j}^j)\in F''$, we have $\underset{\lambda\in \Lambda}{\bigsqcup} \Box \underset{j\in J}{\bigsqcap}  a_{\lambda(j)}^j\in F''$. But we concluded above that for each $\lambda\in \Lambda$ we have $\neg \Box  \underset{j\in J}{\bigsqcap} a_{\lambda(j)}^j\in F''$, which implies $\underset{\lambda\in\Lambda}{\bigsqcap} \neg \Box  \underset{j\in J}{\bigsqcap} a_{\lambda(j)}^j\in F''$ and hence $\neg\underset{\lambda\in\Lambda}{\bigsqcup}  \Box  \underset{j\in J}{\bigsqcap} a_{\lambda(j)}^j\in F''$, contradicting the fact that $F''$ is proper. Thus, we conclude that $H'$ is proper. Hence $H'\sqsubseteq H$, and $H'\in\mathcal{R} (F'')$ by construction.

\item epistemic factivity: $F\in \mathcal{K} (F)$.\\
Immediate from \textbf{factivity of knowledge} and the definition of $\mathcal{K} $.
\item doxastic consistency: $\mathcal{B} (F)\neq\varnothing$.

We claim that the smallest filter $H$ extending $\{a_1\sqcup\dots\sqcup a_n\mid B a_1\sqcup\dots\sqcup B a_n\in F\}$ is proper. If not, then there is a family $\{a_1^j\sqcup\dots\sqcup a_{n_j}^j\}_{j\in J}$ for a nonempty finite $J$ such that $B a_1^j\sqcup\dots\sqcup B a_{n_j}^j\in F$ for each $j\in J$ and $\underset{j\in J}{\sqcap} (a_1^j\sqcup\dots\sqcup a_{n_j}^j)=0$. As above, let $\Lambda$ be the set of all choice functions $\lambda$ such that for $j\in J$,  ${1\leq \lambda(j)\leq n_j}$. It follows that for each $\lambda\in\Lambda$, $\underset{j\in J}{\bigsqcap} a_{\lambda(j)}^j=0$. Then  $\underset{j\in J}{\bigsqcap} (B a_1\sqcup\dots\sqcup B a_n)\leq \underset{\lambda\in\Lambda}{\bigsqcup}\,B \underset{j\in J}{\bigsqcap}  a_{\lambda(j)}^j \leq 0$ using the same reasoning as in (\ref{DistAglom}) for the first inequality and  \textbf{consistency of belief} for the second. Thus, $0\in F$, contradicting the fact that $F$ is proper.  Hence $H$ is proper, and by construction, $H\in \mathcal{B} (F)$.

\item doxastic inclusion: $\mathcal{B} (F)\subseteq \mathcal{K} (F)$.\\
Immediate from \textbf{knowledge-belief entailment} and the definitions of $\mathcal{B} $ and $\mathcal{K} $.
\end{itemize}
This completes the proof of part \ref{FilterRep1}.

For part \ref{FilterRep2}, we already observed that $\widehat{\cdot}$ preserves meet and complement in (i) and (ii) in the second paragraph of the proof, so next we prove that $\widehat{\cdot}$ preserves the awareness operation. Suppose $F\in\widehat{ A  b}$, so $ A  b\in F$. Hence $ A \neg b\in F$ by the \textbf{symmetry} axiom. Now consider any $F'\sqsubseteq F$, so $F'\supseteq F$, and suppose $H\in \mathcal{A} (F' )$. Hence $H=\mathord{\Uparrow}c$ for some $c$ such that $ A  c\in F'$. Then since $ A  b\in F'$ and $ A \neg b\in F'$, we have $ A  (b\sqcap c)\in F'$ and $ A  (\neg b\sqcap c)\in F'$ by  \textbf{agglomeration}. For $d\in \{b,\neg b\}$, if $d\sqcap c\neq 0$, then from $ A  (d\sqcap c)\in F'$ we have $\mathord{\Uparrow} (d\sqcap c)\in\mathcal{A}  (F')$ and hence $\mathrm{max}(\widehat{d}\cap \mathord{\downarrow}H)=\{\mathord{\Uparrow} (d\sqcap c)\}\subseteq \mathcal{A}  (F')$; on the other hand, if $d\sqcap c = 0$, then $\mathrm{max}(\widehat{d}\cap \mathord{\downarrow}H)=\mathrm{max}(\varnothing)=\varnothing\subseteq \mathcal{A} (F')$. Hence  $\mathrm{max}(\widehat{b}\cap \mathord{\downarrow}H) \cup \mathrm{max}(\neg \widehat{b}\cap \mathord{\downarrow}H)\subseteq \mathcal{A} (F')$, which shows $F\in \mathbf{A}  \widehat{b}$. Now suppose $F\not\in \widehat{ A  b}$, so $ A  b\not\in F$. Then $b\neq 0$ by the \textbf{tautology}  and \textbf{symmetry} axioms, and  $\mathord{\Uparrow}b\not\in \mathcal{A} (F)$. Then since $\{1\}\in \mathcal{A} (F)$ by the \textbf{tautology} axiom and $\mathrm{max}(\widehat{b}\cap \mathord{\downarrow}\{1\})= \{\mathord{\Uparrow}b\}$, it follows that $F\not\in \mathbf{A}  \widehat{b}$. This completes the proof that $\widehat{ A  b}=\mathbf{A} \widehat{b}$.

It is now easy to see that the frame is standard: for if $H\in \mathcal{A} (F)$, then $H$ is the principal filter of a nonzero element $a$ such that $A  a\in F$, in which case by the previous paragraph, $F\in \mathbf{A} \widehat{a}$, so $F\in \mathbf{A} \mathord{\downarrow}H$.

Finally, we show that $\widehat{\cdot}$ preserves the knowledge operation (the proof for belief is analogous). Suppose $F\in\widehat{K b}$, so $K  b\in F$. Then by \textbf{knowledge-awareness entailment}, $A  b\in F$, so by our reasoning above, $F\in \mathbf{A} \widehat{b}$. Moreover, $K  b\in F$  implies that for each $H\in\mathcal{K} (F)$, we have $b\in H$ and hence $H\in\widehat{b}$, so $\mathcal{K} (F)\subseteq\widehat{b}$. Thus, by definition of $\mathbf{K} $, $F\in \mathbf{K} \widehat{b}$. Now suppose $F\not\in\widehat{K b}$, so $K  b\not\in F$ and hence $b\neq 1$ by \textbf{necessitation}. Case 1: $A b\not\in F$. Then by our reasoning above, $F\not\in \mathbf{A} \widehat{b}$, which implies $F\not\in \mathbf{K} \widehat{b}$. Case 2: $A b\in F$. Then we claim that the smallest filter $H$ extending
\[\{a_1\sqcup\dots\sqcup a_n\mid K a_1\sqcup\dots\sqcup K a_n\in F\}\cup \{\neg b\}\] 
is proper. If not, then there is a family $\{a_1^j\sqcup\dots\sqcup a_{n_j}^j\}_{j\in J}$ for a nonempty finite $J$ with $K a_1^j\sqcup\dots\sqcup K a_{n_j}^j\in F$ for each $j\in J$ and $\underset{j\in J}{\sqcap}(a_1^j\sqcup\dots\sqcup a_{n_j}^j)\leq  b$. As before,  let $\Lambda$ be the set of all choice functions $\lambda$ such that for $j\in J$,  $1\leq \lambda(j)\leq n_j$. Then for each $\lambda\in\Lambda$, $\underset{j\in J}{\bigsqcap} a_{\lambda(j)}^j\leq  b$ and hence

\begin{equation}A  b\sqcap \underset{j\in J}{\bigsqcap} K  a_{\lambda(j)}^j\leq A   b\sqcap K  \underset{j\in J}{\bigsqcap} a_{\lambda(j)}^j \leq K   b,\label{KnowledgeEqs}\end{equation}
using  \textbf{knowledge agglomeration} for the first inequality and \textbf{awareness-restricted monotonicity} for the second. Then
\begin{equation*} A   b\sqcap \underset{j\in J}{\bigsqcap} (K  a_1^j\sqcup\dots\sqcup K  a_{n_j}^j)\leq A   b\sqcap \underset{\lambda\in \Lambda}{\bigsqcup}\, \underset{j\in J}{\bigsqcap} K  a_{\lambda(j)}^j\leq \underset{\lambda\in \Lambda}{\bigsqcup}(A   b\sqcap \underset{j\in J}{\bigsqcap} K  a_{\lambda(j)}^j ) \leq K   b ,\label{DistStep2}\end{equation*}
using Boolean distributivity for the first two inequalities and then (\ref{KnowledgeEqs}) for the last inequality. Then since the element on the far left belongs to $F$, so does $K  b$, contradicting what we derived above. Thus, $H$ is indeed a proper filter. By construction $H\in \mathcal{K} (F)$, and $b\not\in H$ since $H$ is proper, so $H\not\in\widehat{b}$. Hence $F\not\in\mathbf{K} \widehat{b}$.\end{proof}

\bibliographystyle{plainnat}
\bibliography{awareness}

\end{document}